\documentclass[11pt]{article}

\usepackage{natbib}
\usepackage{graphicx}
\usepackage{float}
\usepackage{subcaption}
\usepackage{enumitem}
\usepackage{amssymb}
\usepackage{epstopdf}
\usepackage{amsmath}
\usepackage{amsthm}
\usepackage{verbatim}
\usepackage{multirow,bigdelim}
\usepackage{rotating}
\usepackage{amssymb}
\usepackage{placeins}
\usepackage{tikz}
\usetikzlibrary{positioning,shapes.geometric}
\usepackage[margin=1.5in]{geometry}

\makeatletter
\newenvironment{subtheorem}[1]{%
  \def\subtheoremcounter{#1}%
  \refstepcounter{#1}%
  \protected@edef\theparentnumber{\csname the#1\endcsname}%
  \setcounter{parentnumber}{\value{#1}}%
  \setcounter{#1}{0}%
  \expandafter\def\csname the#1\endcsname{\theparentnumber\:\alph{#1}}%
  \ignorespaces
}{%
  \setcounter{\subtheoremcounter}{\value{parentnumber}}%
  \ignorespacesafterend
}
\makeatother
\newcounter{parentnumber}

\newtheorem{prop}{Proposition}
\newtheorem{Lem}{Lemma}
\newtheorem{Ass}{Assumption}
\newtheorem{RegAss}{Regularity Assumption}
\newcommand*\rot{\rotatebox{90}}
\newcommand{\ci}{\mbox{\protect $\: \perp \hspace{-2.3ex}\perp$ }}
\newcommand{\E}{{\mbox{E}}}
\newcommand{\Var}{{\mbox{Var}}}
\newcommand{\Cov}{{\mbox{Cov}}}
\newcommand{\N}{{\mbox{N}}}
\newcommand{\z}{{\bf z}}
\newcommand{\x}{{\bf x}}
\newcommand{\X}{{\bf X}}
\newcommand{\y}{{\bf y}}
\newcommand{\bgamma}{{\bf \gamma}}
\newcommand{\bbeta}{{\bf \beta}}

\newcommand{\blambda}{\boldsymbol \lambda}

\newcommand{\limes}{\lim \limits_{n\to\infty}}
\newcommand{\bias}{\mbox{bias}}

\title{Causal inference taking into account unobserved confounding}

\author
{\bf{Minna Genb\"ack}  \\
CEDAR, Ume\aa \, University, 901 87, Ume\aa, Sweden. \\
\textit{email:} minna.genback@umu.se\\
\\
\bf{and}\\
\\
\bf{Xavier de Luna} \\
Department of Statistics, USBE, Ume\aa \, University, 901 87, Ume\aa, Sweden.\\
\textit{email:} xavier.de.luna@umu.se}

\begin{document}
\date{}
\maketitle
\begin{abstract}
\noindent Causal inference with observational data can be performed under an assumption of no unobserved confounders (unconfoundedness assumption). There is, however, seldom clear subject-matter or empirical evidence for such an
assumption. We therefore develop uncertainty intervals for average causal effects based on outcome regression estimators and doubly robust estimators,
which provide inference taking into account both sampling variability and uncertainty due to unobserved confounders. In contrast with sampling variation, uncertainty due unobserved confounding does not decrease with increasing sample size.
The intervals introduced
are obtained by deriving the bias of the estimators due to unobserved confounders. We are thus also able to contrast the size of the bias due to violation of the unconfoundedness assumption, with bias due to misspecification of the models used to explain potential outcomes. 
This is illustrated through numerical experiments where bias due to moderate unobserved confounding dominates misspecification bias for typical situations in terms of sample size and modeling assumptions. We also study the empirical coverage of the uncertainty intervals introduced and apply the results to a study of the effect of regular food intake on health. An R-package implementing the inference proposed is available.
\end{abstract}

%

\textit{Keywords:} Average treatment effect; double robust; ignorability assumption; regular food intake; sensitivity analysis; uncertainty intervals.



%

\section{Introduction}
In observational studies, a causal effect of a treatment can be identified given the assumption that all variables confounding the effect on the outcome of interest are observed. This unconfoundedness assumption (also called ignorability of the treatment assignment mechanism) is regarded as the Achilles heel of non-experimental studies (\citealp{Liu:2013}) and is not testable without further information (e.g., \citealp{Luna:2006,LunaJ:2014}). 
We
therefore develop in this paper uncertainty intervals (\citealp{Vansteelandt:2006}) for average causal effects based on outcome regression estimators and doubly robust estimators,
which provide inference taking into account sampling variability and uncertainty due to unobserved confounders.
The intervals
are obtained by deriving the bias of the estimators due to unobserved confounders as a function of
a parameter $\rho$ (called bias parameter in the sequel) quantifying the amount of unobserved confounding. Using the bias expressions, we deduce bounds on the average causal effects (an identification set in contrast to point identification available under unconfoundedness). Combining these bounds with sampling variability yields uncertainty intervals that have the property to cover the parameter of interest with higher probability than an a priori chosen level (say 95\%). 
The bounds obtained are useful when information on the parameter $\rho$ is available since they can then be made tighter, in contrast with worst case scenario bounds (e.g., \citealp{M:03}, \citealp{HM:06}). 

The approach taken here is directly related to the quickly expanding literature on methods to perform a sensitivity analysis to the  unconfoundedness assumption (\citealp{R:10}, Chap. 14). And, indeed, the uncertainty intervals proposed may be used to perform such a sensitivity analysis whereby the maximum value of the bias parameter is presented for which the uncertainty interval covers zero (no causal effect). Among existing methods to perform sensitivity analyses, many are based on specifying parametric models on how a potential confounder affects the outcome and treatment assignment given the observed covariates, thereby introducing bias parameters (one for the effect of the confounder on the observed outcome and the other for the effect of the confounder on the treatment). Then, typically, using some distributional assumptions for the hypothetical unobserved confounder, the latter is integrated out in order to obtain the bias of an estimator as function of the bias parameters;  see, e.g., \cite{R:10,Lin:1998,Robins:2000,Rosenbaum,I:03,VanderWeele:2011} using a frequentist approach, and \cite{Greenland:2005}; \cite{Luna:2014} using a Bayesian framework. The confounder and outcome are often assumed binary, but some approaches allow for a continuous confounder and/or outcome (e.g., \citealp{VanderWeele:2011}). A directly related literature deals with sensitivity analyses to departures from the ignorability assumption of a missing outcome data mechanism (missing at random assumption); see, e.g., \cite{Copas:01,CE:05,SDR:03, DH:08}. In fact, by using the potential outcome framework (\citealp{Rubin:1974}), the estimation of a causal effect can be cast into a problem of missing outcome (unobserved potential outcomes) and sensitivity analyses for the missing at random assumption can readily be used to study deviations from the unconfoundedness assumption. Alternative approaches to parametrising the relation between a potential unobserved confounder and the outcome and treatment, is to define the bias due to non-ignorability (treatment assignment/missingness mechanism) as the bias parameter, and, e.g., put a prior on this bias within a Bayesian framework (\citealp{DH:08, Josefsson:2016}). Finally, an approach we find appealing from an interpretation point of view is to consider as the bias parameter, the correlation (induced by unobserved confounders) between the treatment assignment and the potential outcomes given the observed covariates; see \cite{CopasLi} and \cite{Genback:2015} within a missing outcome context, and
\cite{Imai:2010} within a parametric mediation analysis context. This approach has the advantage of introducing only one bias parameter for each missingness mechanism.

In this paper, we build upon the latter alternative to perform inference on a causal parameter that takes into account uncertainty due to unobserved confounding and sampling variability. When estimating the average causal effect, two data missingness mechanisms must be considered (one for the outcome under treatment and one for the outcome when no treatment is assigned) implying the need for two ignorability assumptions for point-identification, i.e. two bias parameters. On the other hand, if the interest lies solely in an average treatment effect on the treated (or the non-treated), then only one missingness mechanism has to be dealt with, thereby only one bias parameter. We obtain bounds on the causal effect of interest by deducing the bias of the estimators as function of the bias parameter(s). Thus, we are also able to contribute by contrasting the size of the bias of the outcome regression estimator due to, i) violation of the unconfoundedness assumption, and to ii) the misspecification of the models used to explain outcome. Indeed three types of uncertainty can be distinguished: sampling variation, model misspecification and unobserved confounding. Sampling variation decreases as sample size increases. Model misspecification bias may be tackled with double robust estimation, and in any case this bias can in principle be made arbitrarily small with larger samples, e.g. under sparsity assumptions, by increasing the flexibility of the models used. On the other hand, bias due to unobserved confounding, does not disappear with increasing sample size as long as unobserved confounders are omitted, and is therefore essential to take into account in observational studies. 

The paper is organized as follows. First, a framework for deducing bounds based on a parametrised model is introduced in Section 2. In Section 3 we focus on outcome regression and double robust estimators of average causal effects. We deduce their bias under confounding and show that confounding bias and model misspecification bias are separable. From confounding bias expressions we obtain bounds, and their corresponding uncertainty intervals for the parameter of interest. The R-package \texttt{ui} implements the methods proposed (available at \texttt{http://stat4reg.se/software}).  In Section 4 simulated experiments are conducted to study the relative size of the biases due to confounding and model misspecification, as well to investigate empirical coverage of the uncertainty intervals proposed. In Section 5 we perform a sensitivity analysis in a real data example. The paper is concluded in Section 6.

\section{Identification and sampling variation}

\subsection{Model for point identification}

Let $y^1$ and $y^0$ be two potential outcomes, where $y^1$ is the outcome when treatment $1$ is assigned ($z=1$), and $y^0$ is the outcome when treatment $0$ is assigned ($z=0$). The two potential outcomes are defined for each individual in the study although only one is observed ($y^1$ is observed when $z=1$, and $y^0$ is observed when $z=0$).  We further assume that a set of covariates $x$ is observed for all individuals. To allow for a more compact notation in the following sections we let the first element of $x$ represent the intercept. The evaluation of a treatment effect on the outcome may be done by considering average effects. In this paper we focus on both $E(y^1-y^0)=\tau$, the average causal effect, and $E(y^1-y^0\mid z=1)=\tau^1$, the average causal effect on the treated.

Without loss of generality, let us write
\begin{equation}\label{outcome.mod}
	y^0=f^0(x)+\varepsilon^0,\ \ \ \ \ y^1=f^1(x)+\varepsilon^1, 
	\end{equation}
where $f^j(x)$ $j=0,1$, are functions of $x$ and $E(\varepsilon^0\mid x)=E(\varepsilon^1\mid x)=0$. 
Let further
\begin{equation}\label{assignment.mod}
z^*=g(x)+\eta,\ \mbox{and}\ \ \  z={\bf I}(z^*>0),
\end{equation}
where $\bf I(\cdot)$ is an indicator function, $z^*$ is not observed, $g(x)$ is a function of $x$ and $E(\eta\mid x)=0$.

We can now give sufficient conditions for point identification of $\tau^1$ and $\tau$.
\begin{subtheorem}{Ass}
\label{Ass1}
\begin{Ass}
\label{Ass1a}
 $\varepsilon^0\ci \eta\mid x$.
\end{Ass}
\begin{Ass}
\label{Ass1b}
$\Pr(z=0\mid x)>0, \forall x \in \mathcal{X}$, where $\mathcal{X}$ is the support of $x.$
\end{Ass}
\end{subtheorem}

\begin{subtheorem}{Ass}
\label{Ass2}
\begin{Ass}
\label{Ass2a}
$\varepsilon^1\ci \eta\mid x$.
\end{Ass}
\begin{Ass}
\label{Ass2b}
$\Pr(z=1\mid x)>0, \forall x \in \mathcal{X}$, where $\mathcal{X}$ is the support of $x.$
\end{Ass}
\end{subtheorem}

Assumption \ref{Ass1a} and \ref{Ass2a} are often called unconfoundedness or ignorability assumptions. We have that $\tau^1$ and $\tau$ are point identified under Assumption \ref{Ass1} and under Assumptions \ref{Ass1} and \ref{Ass2} respectively (\citealp{RR:83}).


\subsection{Bias parameters and uncertainty intervals}
\label{Identif.sec}
The ignorability assumptions of the treatment assignment mechanism cannot be tested with the observed data unless extra information is available, e.g, instrumental variables; see \cite{LunaJ:2014}. Thus, unless Assumptions \ref{Ass1a} and \ref{Ass2a} are true by design of the study, uncertainty about their validity should be taken into account in the inference.
For this purpose, it is useful to parametrise deviations from Assumptions \ref{Ass1a} and \ref{Ass2a} The following parametrization models realistic deviations and is easy to communicate to potential users.  

\begin{Ass}
\label{Ass3}
Consider model (\ref{outcome.mod}-\ref{assignment.mod}) with $g(x)= g(x; \gamma)$, for an unknown parameter $\gamma$, and, for $j=0, 1$, let $\sigma_j^2=\mbox{Var}(y^j\mid x)<\infty$, $\eta \sim \mbox{N}(0, 1)$, $\varepsilon^j= \rho_j \sigma_j \eta + \xi_j$, $E(\xi_j)=0$, $\Var(\xi_j)=\sigma_j\sqrt{1-\rho^2}$ and $\xi_j \ci \eta$. 
\end{Ass}
Here we have introduced the bias parameters $\rho_0=Corr(\varepsilon^0,\eta)$ and $\rho_1=Corr(\varepsilon^1,\eta)$. 
This model is such that Assumption \ref{Ass1a} holds when $\rho_0=0$ and not otherwise, and Assumption \ref{Ass2a} holds when $\rho_1=0$ and not otherwise. Hence, these parameters describe departures from ignorability of the treatment assignment mechanism. We call them bias parameters since they tune the bias that will result from assuming ignorability. 
The normality of $\eta$ corresponds to a choice of link function for (\ref{assignment.mod}) which is convenient mathematically in the sequel, but is not otherwise essential in the model. Note, morevoer, that normality may be relaxed to a more general class of distributions (\citealp{Genback:2015}, Sec. 3.2).

If we have unmeasured confounders, one way to interpret $\rho_j$ is to rewrite the error terms $\varepsilon^0$, $\varepsilon^1$ and $\eta$ from equations (\ref{outcome.mod}) and (\ref{assignment.mod}) as the sum of the error that can and cannot be explained by the unmeasured confounder. For instance, if we believe that unmeasured confounder(s) explain $a\cdot100\, \%$ of the variation in $\eta$ and $b\cdot100\, \%$ of the variation in $\varepsilon^j$, and that the unmeasured confounders affect treatment assignment negatively and $y^j$ positively, then $\rho_j=(-\sqrt{a})(+\sqrt{b})$. 

The approach proposed here is to deduce an identification interval for the parameter of interest $\tau$ by using an estimator $\hat\tau$ which is unbiased for $\tau$ under unconfoundedness ($\rho_0=\rho_1=0$). Then, the bias of the estimator is computed as a function of the bias parameter $b(\rho_0,\rho_1;\theta)=E(\hat\tau)-\tau$, where $\theta$ is a nuisance parameter vector (containing $f^0$, $f^1$, $\sigma$ and $\gamma$). Finally, this bias expression together with
out-of-data information (if any) on $\rho_0$ and/or $\rho_1$ in the form of an interval $\rho_j \in [\rho_j^L ,\rho_j^U]$ yields an identification interval for $\tau$: 
\begin{equation}\label{ident.set}
	\{\tau: \tau=E_0(\hat\tau)-b(\rho_0,\rho_1;\theta), \rho_0\in [\rho_0^L,\rho_0^U], \rho_1\in [\rho_1^L,\rho_1^U], \theta = \theta_0 \},
\end{equation}
where $\theta_0$ is the true value of $\theta$, and $E_0$ is the expectation taken over the observed data law, i.e. corresponding to the true but unknown values for $\rho_0$ and $\rho_1$. Note that $\rho_j \in [-1, 1]$ is the no out-of-data information case. In some applications, however, one may have out-of-data information, for instance that the treatment assignment is not negatively correlated with the outcome, $\rho_j \in [0,1 -\delta]$, for some $\delta>0$. Another instance arises when a rich and relevant set of covariates $x$ is available, in which case one may believe that $\rho_j \in [-\delta,\delta]$ for $\delta\geq 0$ small.

%

In situations where a consistent and asymptotically normal estimator of $E(\hat\tau)-b(\rho_0,\rho_1;\theta)$ is available, denoted $\hat\tau-\hat b(\rho_0,\rho_1)$, with corresponding standard errors, $s.e.(\hat\tau-\hat b(\rho_0,\rho_1))$, an uncertainty interval containing $\tau$ with probability at least $1-\alpha$ (\citealp{Tanja}) is given by:
\begin{equation}
UI(\tau; [\rho_0^L,\rho_0^U], [\rho_1^L,\rho_1^U],\alpha)=\bigcup_{\rho_0 \in [\rho_0^L,\rho_0^U],\rho_1\in [\rho_1^L,\rho_1^U]} CI(\tau; \rho_0,\rho_1,\alpha),
\label{UI.eqn}
\end{equation}
where
\begin{equation*}
CI(\tau; \rho_0,\rho_1,\alpha)=\left(  \hat\tau-\hat b(\rho_0,\rho_1) - c_{\frac{\alpha}{2}}s.e.(\hat\tau-\hat b(\rho_0,\rho_1)), \hat\tau-\hat b(\rho_0,\rho_1) + c_{\frac{\alpha}{2}}s.e.(\hat\tau-\hat b(\rho_0,\rho_1)) \right)
\label{CI.eqn}
\end{equation*}
and $c_{\frac{\alpha}{2}}$ is the $(1-\alpha/2)100 \%$ percentile of the standard normal distribution.

\section[Identification intervals associated to outcome regression and doubly robust estimators]{Outcome regression and doubly robust estimators: bias and inference} 
\label{Bias.sec}
 We now consider two families of estimators and apply the approach described above to deduce their bias and the resulting uncertainty intervals. We will use the following assumption of correctly specified regression models.

\begin{subtheorem}{Ass}
\label{Ass4}
\begin{Ass}
$f^0(x)=\beta^{0'}x$,
\label{Ass4a}
\end{Ass}
\begin{Ass}
$f^1(x)=\beta^{1'}x$,
\label{Ass4b}
\end{Ass}
\end{subtheorem}
where $\beta^0$ and $\beta^1$ are parameter vectors, and the first element of $x$ is 1. Assumption 4 can be made very general by replacing $\beta' x$ by $\tilde \beta' w$ where $w$ includes bases functions of the space spanned by $x$, e.g. cubic splines. 

\subsection{Estimators of average causal effects}
\label{Estimators}

 Let us assume that we have a random sample of size $n$ of which $n_1$ are treated and $n_0$ are controls (not treated), and let $\mathcal{I}_1$ be the indexes for the treated and $\mathcal{I}_0$ be the indexes for the controls. We denote $\hat\beta^j_{OLS}=(\X_j'\X_j)^{-1}\X_j' \y_j$, for $j=0,1$, where $\X_j$ is a matrix of size $n_j \times (p+1)$ containing the elements $\{x_i' ; i \in \mathcal{I}_j\} $, $p$ is the number of covariates and $\y_j$ is a vector with elements $\{y_i ; i \in \mathcal{I}_j\}$. We consider the following outcome regression estimators for the average causal effect  $\tau$ and average causal effect on the treated $\tau^1$ (e.g., \citealp{Tan:2007}):
 \begin{equation}
\hat\tau^1_{OR} = \frac{1}{n_1}\sum \limits_{i =1}^n z_i\left(y_i-\hat\beta^{0'}_{OLS}x_i\right),
\end{equation}
\begin{equation}\hat\tau_{OR}=\frac{1}{n}\sum_{i=1}^{n}\left( \hat\beta^{1'}_{OLS}x_i-\hat\beta^{0'}_{OLS}x_i\right).
\end{equation}
The outcome regression estimator is an unbiased estimate of $\tau^1$ (or $\tau$) under Assumption \ref{Ass1} (and \ref{Ass2}), and Assumption \ref{Ass4a} (and b.). The doubly robust estimator consists of the outcome regression estimators above with a correction term for the potential misspecification of $f^j(x)$. For $\tau^1$ and $\tau$, the doubly robust estimators are (e.g., \citealp{SRR:99}, \citealp{Lunceford:2004} and \citealp{RF:13}):
\begin{eqnarray}\label{DR1.eq}
\hat{\tau}^1_{DR}&=&\hat\tau^1_{OR}-\frac{1}{n_1}\sum \limits_{i =1}^n(1-z_i)\frac{y_i-\hat\beta^{0'}_{OLS}x_i}{1-\hat p(x_i)} ,
\end{eqnarray}
\begin{equation}\label{DR.eq}
\hat{\tau}_{DR}=\hat \tau_{OR}
+\frac{1}{n}\sum  \limits_{i =1}^n z_i \frac{y_i-\hat\beta^{1'}_{OLS}x_i}{\hat{p}(x_i)}
-\frac{1}{n}\sum  \limits_{i =1}^n (1-z_i)\frac{y_i-\hat\beta^{0'}_{OLS}x_i}{1-\hat{p}(x_i)},
\end{equation}
where $\hat{p}(x_i)$ is an estimate of the propensity score $p(x_i)=\Pr(z=1\mid x_i)$. The doubly robust estimators are unbiased under Assumption \ref{Ass1} (and \ref{Ass2}), and Assumption \ref{Ass4a} (and b.) and/or a correctly specified propensity score model, see below.

\subsection{Bias expressions}
\label{Bias.section}

For the sake of simplicity we denote the (total) bias of an estimator $\hat \tau$ by bias$_T(\hat \tau)=b(\hat \tau, \rho_0,\rho_1;\theta)$. We investigate two sources of bias, the bias due to model misspecification (Assumption \ref{Ass4} not fulfilled), bias$_M$, and the bias due to unobserved confounding (non-ignorability of the treatment assignment mechanism, Assumption \ref{Ass3}), bias$_C$, as summarized in Table \ref{tab:bias.overview}. All proofs are given in  Appendix B.

\begin{prop} \emph{Bias of OR estimators under correctly specified models.}\\ 
Under Assumptions \ref{Ass1b}, \ref{Ass3}, \ref{Ass4a} and Regularity Assumption \ref{RegAss1} in Appendix A,
	\begin{align*}
	\limes \emph{bias}_T (\hat \tau_{OR}^1)&=\limes\emph{bias}_C (\hat \tau_{OR}^1) =\\
	&=\rho_0\sigma_0 \left(E(\lambda_1(g(x; \gamma))\mid z=1)+ E((\X_0'\X_0)^{-1}\X_0' \boldsymbol \lambda_0) E(x\mid z=1)\right),
	\end{align*}
 and under Assumptions \ref{Ass1b}, \ref{Ass2b}, \ref{Ass3} and \ref{Ass4},
		\begin{align*}
 \emph{bias}_T (\hat \tau_{OR})= \emph{bias}_C (\hat \tau_{OR})= \frac{1}{n}\E\left[ {\bf1}_n \X  \left(\rho_1\sigma_1(\X_1'\X_1)^{-1}\X_1' \boldsymbol\lambda_1 +\rho_0\sigma_0(\X_0'\X_0)^{-1}\X_0' \boldsymbol 	\lambda_0 \right)\right],
 \end{align*}
where ${\bf1}_n$ is a vector with all elements 1 of length $n$, and $\boldsymbol \lambda_j$ is a vector of length $n_j$ containing the elements $\{\lambda_j(g(x_i; \gamma)) ; i \in \mathcal{I}_j\} $, $\lambda_j$ is the inverse Mill's ratio $\lambda_0(g(x_i; \gamma))=\frac{\phi(g(x_i; \gamma))}{1-\Phi(g(x_i; \gamma))}$ and $\lambda_1(g(x_i; \gamma))=\frac{\phi(g(x_i; \gamma))}{\Phi(g(x_i; \gamma))}$,  and $\phi$ and $\Phi$ are the normal pdf and cdf. 
\label{Prop.OR1}
\end{prop}
\noindent Let us further investigate model misspecification bias in combination with non-ignorability of treatment. 

\begin{prop} \emph{Bias of OR estimators under model misspecification.}\\
Under Assumptions \ref{Ass1b} ,\ref{Ass3} and Regularity Assumption \ref{RegAss1} in Appendix A,
	\begin{align*}
	\limes \emph{bias}_T(\hat \tau^1_{OR}) &=\limes( \emph{bias}_C(\hat \tau^1_{OR})+  \emph{bias}_M(\hat \tau^1_{OR}))=\\
	&= \emph{bias}_C(\hat \tau^1_{OR})+ \E\left[ f^0(x) |z=1\right] - E\left[(\X_0'\X_0)^{-1}\X_0' f^0(x) \right] \E(x|z=1), 
	\end{align*}
 and under Assumptions \ref{Ass1b}, \ref{Ass2b} and \ref{Ass3},
	\begin{align*}
	\emph{bias}_T(\hat \tau_{OR})&= \emph{bias}_C(\hat \tau_{OR})+ \emph{bias}_M(\hat \tau_{OR})=\\
	&=\emph{bias}_C(\hat \tau_{OR})+\E\left(z f^0(x) \right) - \frac{1}{n}\E\left[ {\bf1}_{n_1} \X_1 (\X_0'\X_0)^{-1}\X_0'  \boldsymbol f^0(x)_{|z=0} \right]\\
	&\qquad \qquad \quad \;\;\; - \E\left((1-z) f^1(x)\right) +\frac{1}{n}\E\left[ {\bf1}_{n_0} \X_0 (\X_1'\X_1)^{-1}\X_1'  \boldsymbol f^1(x)_{|z=1} \right],
	\end{align*}
	where $\boldsymbol f^j(x)_{|z=j}$ is a vector of length $n_j$ containing the elements $\{f^j(x_i) ; i \in \mathcal{I}_j\} $, $j=0,1$. 
\label{Prop.OR2}
\end{prop}

\begin{prop} \emph{Bias of DR estimators.}\\ 
Under Assumptions \ref{Ass1b} and \ref{Ass3} and Regularity Assumption \ref{RegAss1} and \ref{RegAss2} in Appendix A,
	\begin{align*}
	\limes \emph{bias}_T (\hat \tau^1_{DR}) =\limes\emph{bias}_C (\hat \tau^1_{DR}) = \rho_0\sigma_0 \frac{\E(\lambda_0(g(x; \gamma)))}{\Pr(z=1)}.
	\end{align*}	
Under Assumptions \ref{Ass1b}, \ref{Ass2b}, \ref{Ass3} and Regularity Assumption \ref{RegAss2} in Appendix A,
	\begin{align*}
	\emph{bias}_T (\hat \tau_{DR})=\emph{bias}_C (\hat \tau_{DR})=  \rho_1\sigma_1 \E(\lambda_1(g(x; \gamma)))+ \rho_0\sigma_0 \E(\lambda_0(g(x; \gamma))).
	\end{align*}
\label{Prop.DR}
\end{prop}

\begin{table}[ht]
\caption{The total bias of the outcome regression and double robust estimators decomposed into bias due to model misspecification (bias$_M$) and bias due to confounding (bias$_C$), see Proposition \ref{Prop.OR1}-\ref{Prop.DR} for details.}
\begin{center}
\begin{tabular}{cc| l| l}
\hline
	&\multicolumn{1}{c}{}&\multicolumn{2}{c}{Assumption 4} \\
	&\multicolumn{1}{c}{Estimator}&\multicolumn{1}{c}{fulfilled}&\multicolumn{1}{c}{not fulfilled}\\
	\hline\hline
\multirow{2}{*}{$\rho_j=0$}&OR& $\bias_T=0$&$\bias_T=\bias_M$\\
	&DR&$\bias_T=0$& $\bias_T=0$\\
	\hline
	\multirow{2}{*}{$\rho_j\neq0$}&OR &$\bias_T=\bias_C$&$\bias_T=\bias_C +\bias_M$ \\
	&DR&$\bias_T=\bias_C$ &$\bias_T=\bias_C$ \\
\hline
\end{tabular}
\end{center}
\label{tab:bias.overview}
\end{table}
Since the doubly robust estimator is the outcome regression estimator with a correction term, there is a link between Proposition \ref{Prop.OR2} and \ref{Prop.DR}:
 \begin{align*}\limes \mbox{bias}_C(\hat \tau^1_{DR}) - \mbox{bias}_C(\hat \tau^1_{OR}) =  \rho_0 \sigma_0  \E\left(\left.\lambda_0(g(x; \gamma)) -   (\X_0' \X_0)^{-1} \X_0' {\boldsymbol \lambda_0 } x \right| z=1 \right),\qquad
 \end{align*}
 and
 \begin{align*} \mbox{bias}_C(\hat \tau_{DR}) - \mbox{bias}_C(\hat \tau_{OR}) &=   \rho_0 \sigma_0  \E\left(\left.\lambda_0(g(x; \gamma)) -   (\X_0' \X_0)^{-1} \X_0' {\boldsymbol \lambda_0 } x \right| z=1 \right)E(z)  \;+\quad \\
& \rho_1 \sigma_1  \E\left(\left.\lambda_1(g(x; \gamma)) -   (\X_1' \X_1)^{-1} \X_1' {\boldsymbol \lambda_1 } x \right| z=0 \right)E(1-z).
 \end{align*}
Since $\lambda_j$ for $j=0,1$ is often close to linear in $x$ (\citealp{P:00}) the expectation of the difference between $\lambda_j(g(x; \gamma))$ and the linear projection of $\lambda_j(g(x; \gamma))$  is small. Hence the difference between the confounding bias of doubly robust and outcome regression estimators is also typically small.

%

\subsection{Uncertainty intervals}
From Proposition \ref{Prop.OR1}-\ref{Prop.DR} identification intervals and uncertainty intervals as defined in Section \ref{Identif.sec} can be derived for $\tau^1$ and $\tau$. The estimation of all elements of $\bias_C$ (Proposition \ref{Prop.OR1}-\ref{Prop.DR}) is straightforward with the exception of $\sigma_j$, $j=0,1$. A consistent estimator of $\sigma_j$ is given by:
\begin{equation}
\hat \sigma_j=\sqrt{\frac{\hat \sigma_{j, OLS}^2}{1+ (-1)^j \frac{\rho_j^2}{n-p} (g(\x_j; \gamma)^T \blambda_j-\blambda_j^T \X_j (\X_j^T \X_j)^{-1}\X_j^T \blambda_j )}},
\label{sigma}
\end {equation}
where ${\hat \sigma_{OLS,j}^2}$ is the residual sample variance from the OLS fit of the potential outcome $j$ (i.e. $\y^j$ against $\X_j$) and $g(\x_j; \gamma)$ is a vector of length $n_j$ containing the elements $\{g(x_i; \gamma)) ; i \in \mathcal{I}_j\} $; a proof of this result for $j=1$ is found in \citet[page 9]{Tanja}.

From Proposition \ref{Prop.OR1}-\ref{Prop.DR}, identification intervals, assuming $\rho \in [\rho^L, \rho^U]$, can be derived by replacing $\sigma_j$ with $\hat \sigma_j$ from (\ref{sigma}). For instance, from Proposition \ref{Prop.DR}, an estimated identification interval for $\tau^1$ is given by:
\begin{equation*}
\left\{  \tau^1 : \tau^1=\hat \tau^1_{DR} -  \rho_0 \hat\sigma_{0}\frac{\hat\E(\lambda_0(g(x; \gamma)))}{\widehat{\Pr}(z=1)}  , \rho_0 \in [\rho_0^L,\rho_0^U]\right\}.
\label{t1DR.IS}
\end{equation*}

Ignoring the sampling variability from $\widehat{\bias}_T$, and noting that $\hat\tau_{DR}^1$ is asymptotically normally distributed (\citealp{Tsiatis:06}), the lower and upper bound of the uncertainty interval, $UI(\tau_{DR}^1; [\rho_0^L,\rho_0^U],\alpha),$ are respectively (see (\ref{UI.eqn})):
\begin{equation}
\min \limits_{\rho_0 \in [\rho_0^L,\rho_0^U]} \left( \hat \tau^1_{DR} -  \rho_0 \hat\sigma_{0}\frac{\hat\E(\lambda_0(g(x; \gamma)))}{\widehat{\Pr}(z=1)} -c_{\frac{\alpha}{2}}s.e.(\hat \tau^1_{DR}) \right),
\label{UI.ex.l}
\end{equation}
and
\begin{equation}
\max \limits_{\rho_0 \in [\rho_0^L,\rho_0^U]} \left( \hat \tau^1_{DR} -  \rho_0 \hat\sigma_{0} \frac{\hat\E(\lambda_0(g(x; \gamma)))}{\widehat{\Pr}(z=1)} +c_{\frac{\alpha}{2}}s.e.(\hat \tau^1_{DR}) \right),
\label{UI.ex.u}
\end{equation}
where $c_{\frac{\alpha}{2}}$ is the $(1-\alpha/2)100 \%$ percentile of the standard normal distribution. Estimated uncertainty intervals for $\tau$ and the outcome regression estimators are obtained similarly. Standard errors for the outcome regression and doubly robust estimators are given in Appendix C. Note that ignoring the sampling variability from $\widehat{\bias}_T$ in (\ref{UI.ex.l}) and (\ref{UI.ex.u}) should have no serious consequences because it is of lower asymptotic order. This is confirmed by the simulation study in Section \ref{Sim.sec}.

\section{Simulation study}
\label{Sim.sec}
The purpose of the simulation experiments is to illustrate the relative sizes of the biases due to model misspecification and confounding, as well as to study the empirical coverages of the proposed uncertainty intervals. The data is generated using four different designs, linear or non-linear, using one or five covariates. For each design, we use two different treatment assignments with different amount of imbalance in the propensity scores between the treated and the non-treated. One with low imbalance (high overlap), L1$=0.18$ and $0.19$, and one with high imbalance (low overlap), L1$=0.37$ and $0.34$. L1 is the area which is not overlaid in a graph with two density histograms of the propensity scores for the treated and untreated, see \citet[equation (5)]{L1}. This measure varies with different bin-size. We use the default of the function {\tt hist} in {\tt R} statistical software on all the propensity scores (both from treated and untreated) to select bin-size. 
In all designs we have about $40\%$ treated, and use a linear model for the treatment assignment mechanism, i.e. $g(x; \gamma)= \gamma' x$ in (\ref{assignment.mod}). For all four designs we use a sample size of 250 and 500, with 10 000 replications and compute uncertainty intervals based on Proposition \ref{Prop.OR1} and \ref{Prop.DR}, i.e. using the outcome regression estimators adjusting for $\bias_C$ but not $\bias_M$ and using the doubly robust estimators adjusting for $\bias_T=\bias_C$. Finally, we let
$$\begin{pmatrix} \eta\\ \varepsilon^0 \\ \varepsilon^1  \end{pmatrix} \sim MVN\left(\begin{pmatrix}0\\0\\0  \end{pmatrix},\begin{pmatrix}
 1 & \rho_0& \rho_1 \\ 
 \rho_0& 1 & \rho_0 \rho_1 \\ 
 \rho_1 & \rho_0 \rho_1 & 1 \\ 
  \end{pmatrix}\right),$$ 
  where $\rho_0=\rho_1 =$ 0, 0.05, 0.1, 0.3, and 0.5. 

In Design A and B we use $\gamma' = [-0.27, 0.3]$ and $\gamma' = [-0.3, 0.65] $, to generate low and high L1 respectively, and $x \sim \N (0,1)$. 
In Design A (linear) we use the outcome equations: $f^0(x)= 0.5 + 0.5x $, and $f^1(x)= 2.5 + 1.5x $. In Design B (non-linear) we use $f^0(x)=h^0(x)$ and $f^1(x)=h^1(x)$, where:

\vspace{8mm}
\hspace{-5mm}
\begin{tabular}{@{}llll@{}}
 \multirow{4}{*}{$h^0(x)=$} \rdelim\{{4}{3mm} &$0.15-x-0.4x^2$ & $x<-1.5$,&  \multirow{9}{*}{\parbox[l]{1em}{\includegraphics[width=2in]{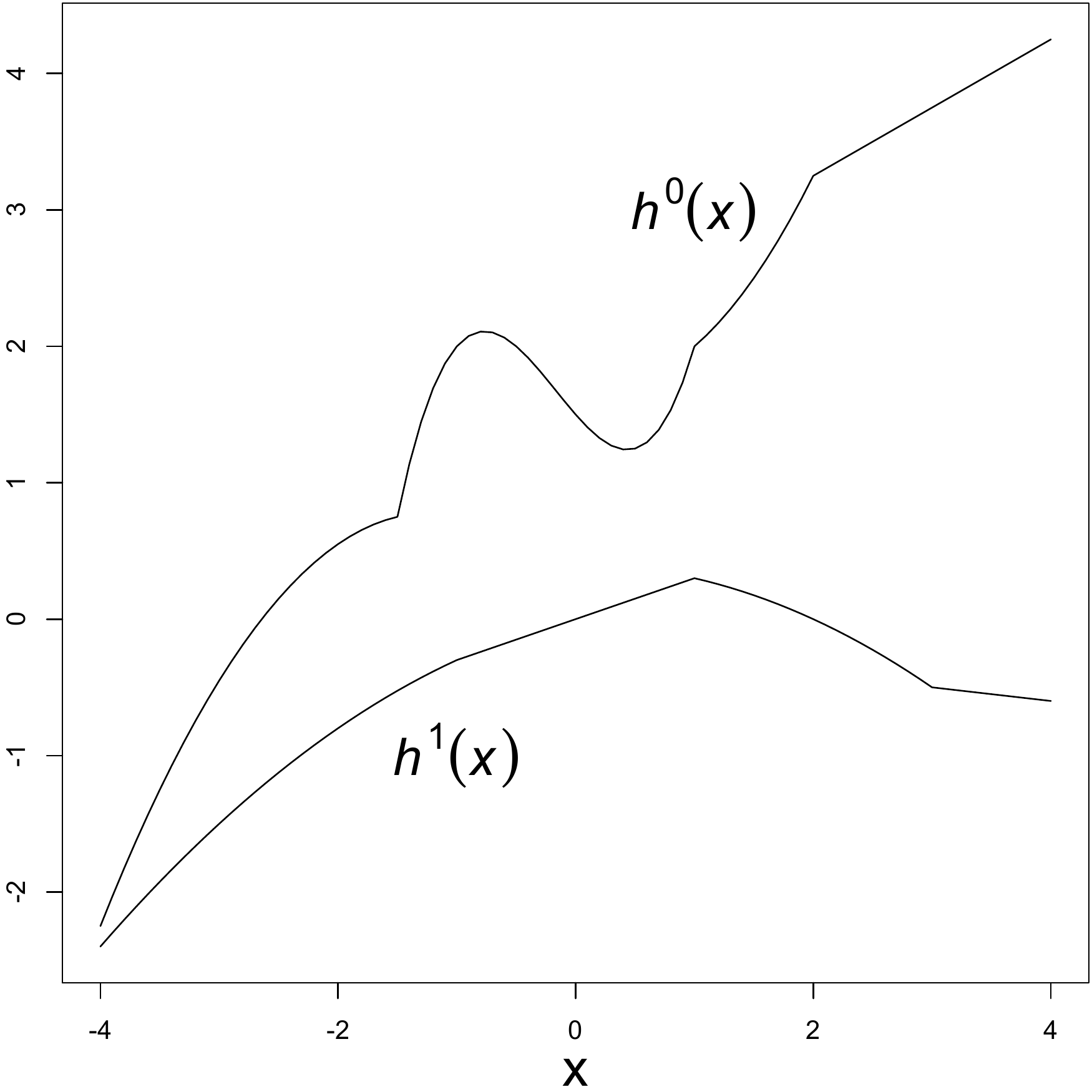}}} \\ 
&$1.5-x+0.5x^2+x^3 $&$-1.5\leq x< 1$, \\ 
&$1.75-0.25x+0.5x^2 $ &$1\leq x< 2$ ,\\
&$2.25+0.5x$&$2<x$,\\
\\
 \multirow{4}{*}{$h^1(x) =$} \rdelim\{{4}{3mm} &$0.2x-0.1x^2$ & $x<-1$, \\ 
&$0.3x $&$-1\leq x< 1$, \\ 
&$0.4-0.1x^2 $ &$1\leq x< 3$, \\
&$-0.2-0.1x$&$3<x$.
\end{tabular}
\vspace{12mm}

\noindent These choices were made to make polynomial approximation difficult. 

In Design C and D we use $\gamma' = (-0.27, 0.2, -0.15, 0.05, 0.15, - 0.1)$ and $\gamma' = (-0.3,0.5, -0.25,\linebreak[0] 0.15, 0.25, - 0.15)$, to generate low and high L1 respectively. The covariates are simulated such that  $x_1 \sim \N(0,1)$, $x_2$ and $x_4$ are Bernoulli distributed with probability $0.5+0.05x_1$ and $0.4+ 0.2x_3 $, $x_3 = 0.015x_1 + u_3,$ where $u_3$ is uniformly distributed in $(-0.5, 0.5)$, and $x_5 =0.04x_1+0.15x_2+0.05x_3+u_5,$ where $u_5 \sim \N(0,1)$. In Design C (linear) we use the outcome equations: $f^0(x)= -0.5 + 0.5x_1  +1.0x_2 + 0.5x_3 -1.0x_4  +1.0x_5 $, and $f^1(x)= 1.5 -1.5x_1 + 4.0 x_2-1.5x_3 + 3.0x_5 $. In Design D (nonlinear) we use, $f^0(x)= h^1(x_1) + 0.1 x_2 -0.3 x_3 -0.6 h^1(x_1)  \cdot x_4 -0.1x_5 $, $f^1(x)= h^0(x_1) +h^0(x_1) \cdot x_2 + 0.3 x_2 -0.2 x_3 -0.4x_4 +0.6x_5$.

For all designs we fit a correctly specified propensity score model, $g(x;\gamma)$, and $f^0$ and $f^1$ with linear models in $x$, i.e. Assumption 4 is fulfilled in Design A and C but not in Design B and D. We compare width and coverage of 95\% confidence intervals for $\tau$ with $UI(\tau; [0,0.2], [0,0.2],0.05)$ and $UI(\tau; [0,0.4], [0,0.4],0.05)$ (the corresponding confidence intervals and uncertainty intervals are used for $\tau^1$).

\subsection{Results}
Figure \ref{Bias.fig} displays magnitude of the empirical bias of the outcome regression estimator (both model misspecification and confouning bias), defined in Section \ref{Bias.section}, for the two non-linear designs with correctly specified propensity score models. We can see that bias$_M(\hat \tau_{OR})$ is larger when the propensity scores of the treated and untreated are more separeted (L1 low). However, even when L1 is high, bias$_M(\hat \tau_{OR})$ and bias$_C(\hat \tau_{OR})$ are of approximately the same magnitude when $\rho_0=\rho_1 = 0.05$, and for $\rho_0=\rho_1 \geq 0.1$ $\bias_C(\hat \tau_{OR})$ dominates $\bias_M(\hat \tau_{OR})$. Note also that in Design D $\bias_C(\hat \tau_{OR})$ and $\bias_M(\hat \tau_{OR})$ have different signs implying that the total bias is smaller than the confounding bias. Hence, the outcome regression estimator has a smaller total bias than the doubly robust estimator in such a case, since the confounding bias of the two estimators is almost the same. 

\begin{figure} [htbp]
\centering
\includegraphics[width=0.75\textwidth]{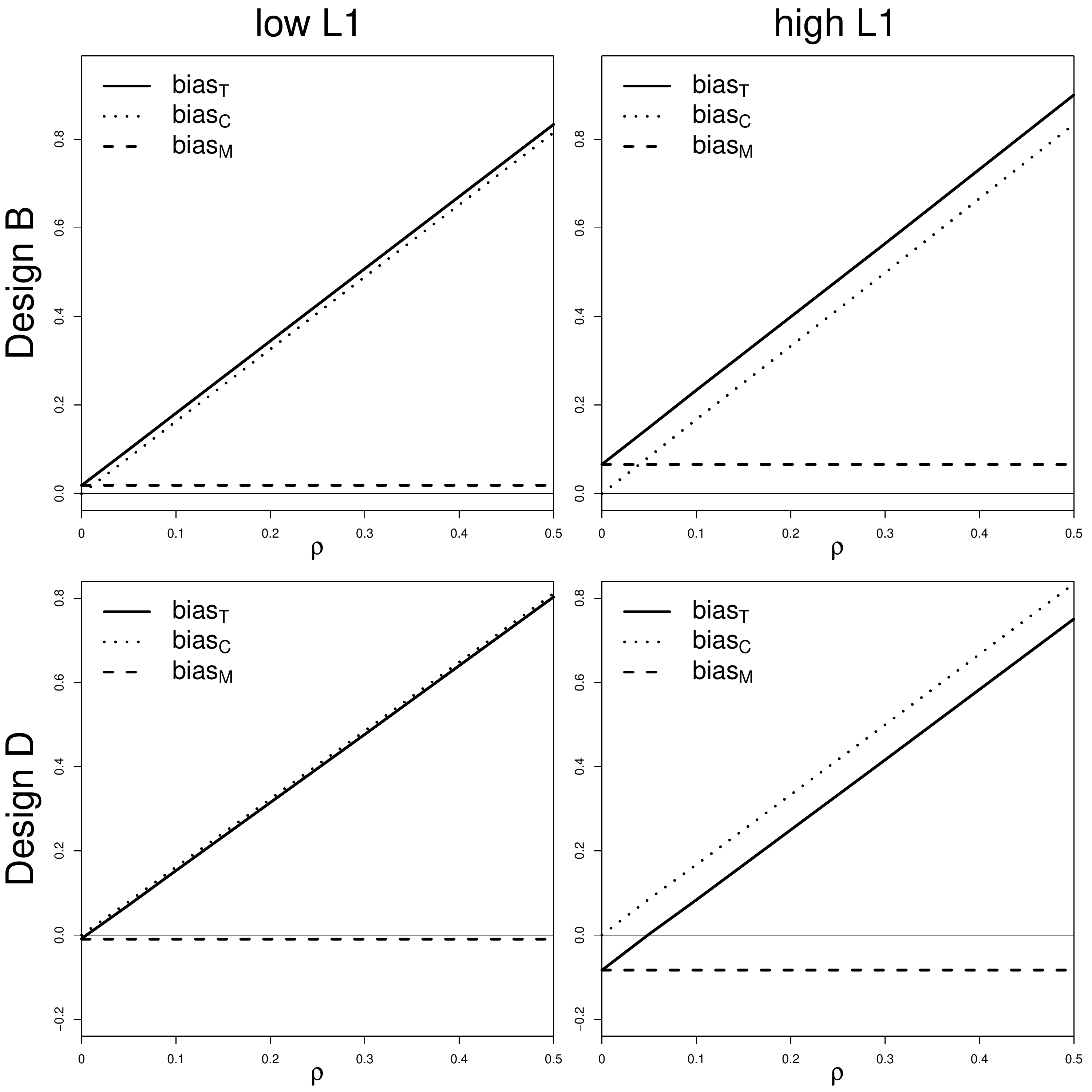}
\caption{The magnitude of $\bias_T(\hat\tau_{OR})$, $\bias_C(\hat\tau_{OR})$ and $\bias_T(\hat\tau_{OR})$ with varying $\rho_0=\rho_1=\rho$ for the two non-linear designs, with low and high imbalance in the propensity scores. }
\label{Bias.fig}
\end{figure}

The uncertainty intervals are wider than the confidence interval per definition, which is confirmed in Figure \ref{t.250.fig} - \ref{t.500.fig} and Appendix D. In particular, the uncertainty intervals derived under the assumption that $\rho_0$ and/or $\rho_1 \in [0,0.4]$ are around twice as wide as the corresponding confidence intervals. The empirical coverage of the 95$\%$ uncertainty intervals are, as expected, generally high if the assumption on $\rho_0$ and/or $\rho_1$ is met ($\rho_0$ and/or $\rho_1$ is covered by the pre-specified interval from which the uncertainty interval is derived); see Figure \ref{t.250.fig} - \ref{t.500.fig} and Appendix D. However, if the assumption is not met the empirical coverage is less than $95 \%$. When using the outcome regression estimator in Design B, we do not necessarily expect 95$\%$ coverage of the UI:s, even if the assumption on $\rho$ is met, because the outcome regression model is misspecified, and $\bias_M$ has the same sign as $\bias_C$. However, the empirical coverage is at least 95$\%$ due to three reasons: first, the uncertainty intervals have higher coverage than 95$\%$ if the models are correctly specified; second, $\sigma_0$ is overestimated due to model misspecification; and third, bias$_M$ has the same size or smaller than bias$_C$. Note finally that the empirical coverage of the 95$\%$ confidence intervals assuming no unobserved confounding is too low even for small $\rho$, see Figure \ref{t.250.fig} - \ref{t.500.fig} and Appendix D.

\begin{figure} [htbp]
\centering
\begin{subfigure}[b]{\textwidth}
\includegraphics[width=\textwidth]{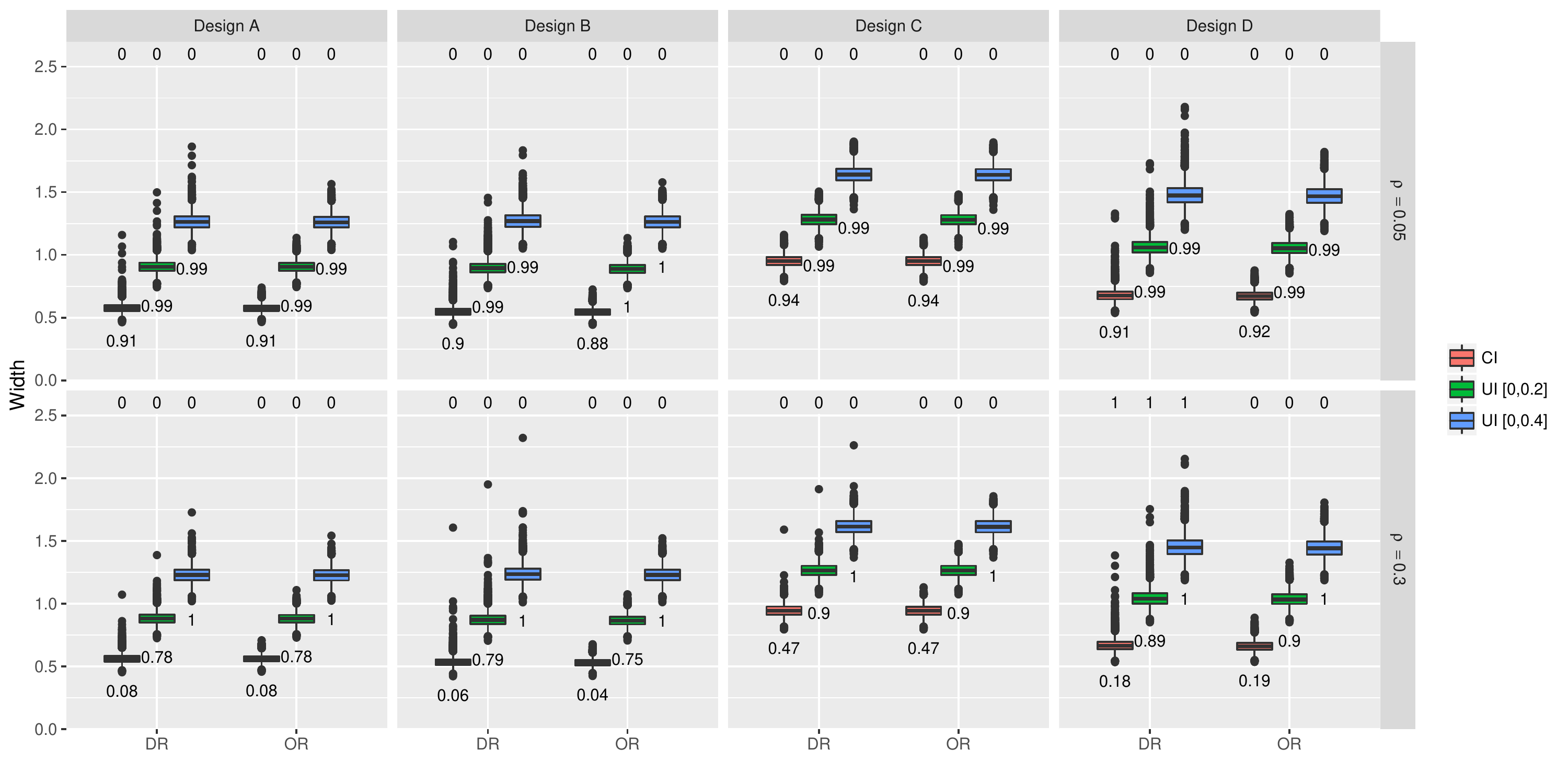}
\caption{Treatment assignment with low L1.}
\label{t.p1.250.fig}
\end{subfigure}
\begin{subfigure}[b]{\textwidth}
\includegraphics[width=\textwidth]{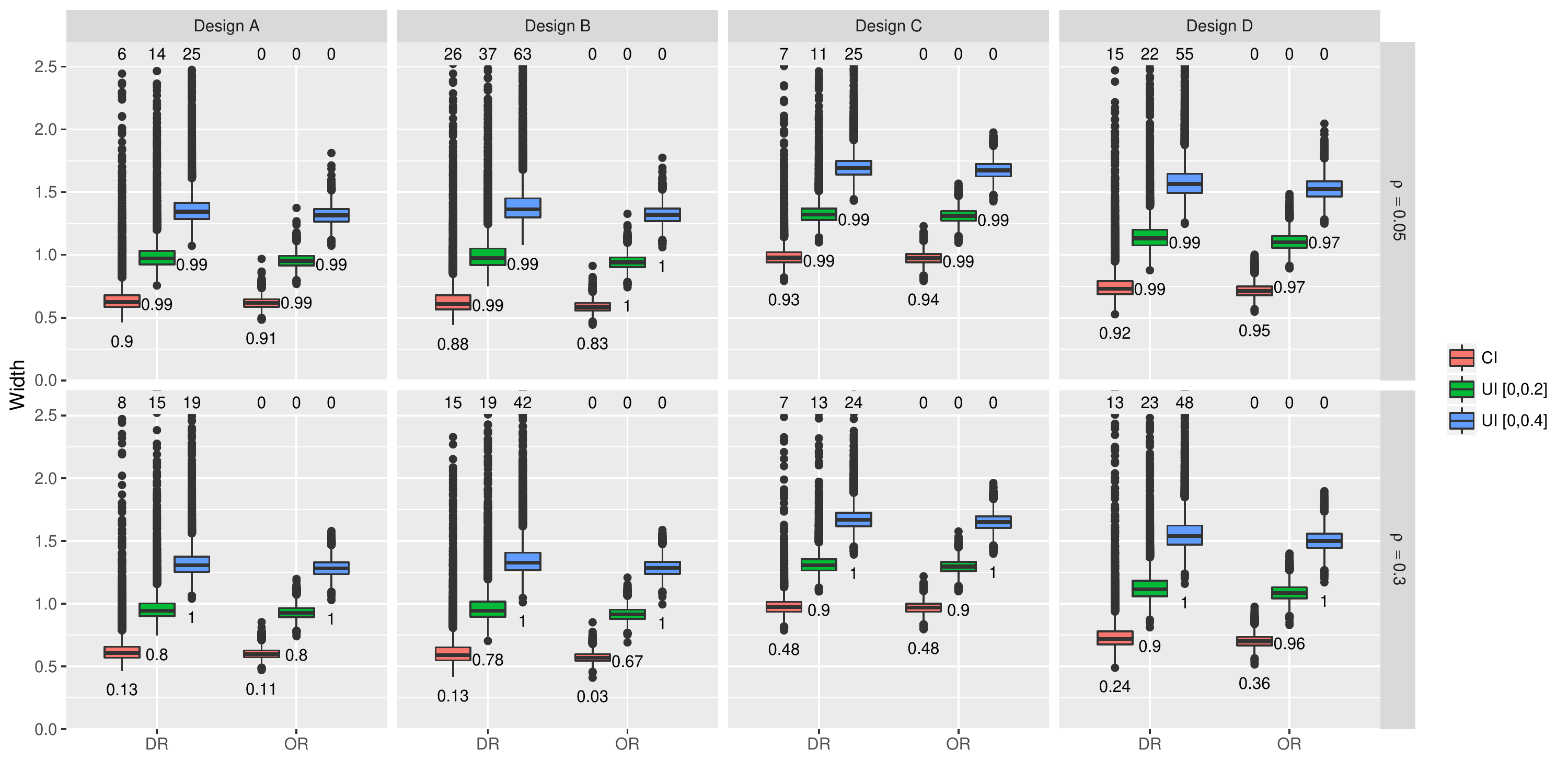}
\caption{Treatment assignment with high L1.}
\label{t.p2.250.fig}
\end{subfigure}
\caption{Boxplot of the width of two 95\% uncertainty intervals (assuming $\rho_j \in [0,0.2]$, green, and $\rho_j \in [0,0.4]$, blue, $j=0,1$) and the 95\% confidence interval, red, for the doubly robust (DR) and outcome regression (OR) estimator of $\tau$ under design A-D for $\rho_0=\rho_1=0.05$ and $0.3$, with sample size 250. The empirical coverage of each interval is written below each boxplot and the number of outliers that lie outside the window is written at the top of the window above each boxplot.}
\label{t.250.fig}
\end{figure}

\begin{figure} [htbp]
\centering
\begin{subfigure}[b]{\textwidth}
\includegraphics[width=\textwidth]{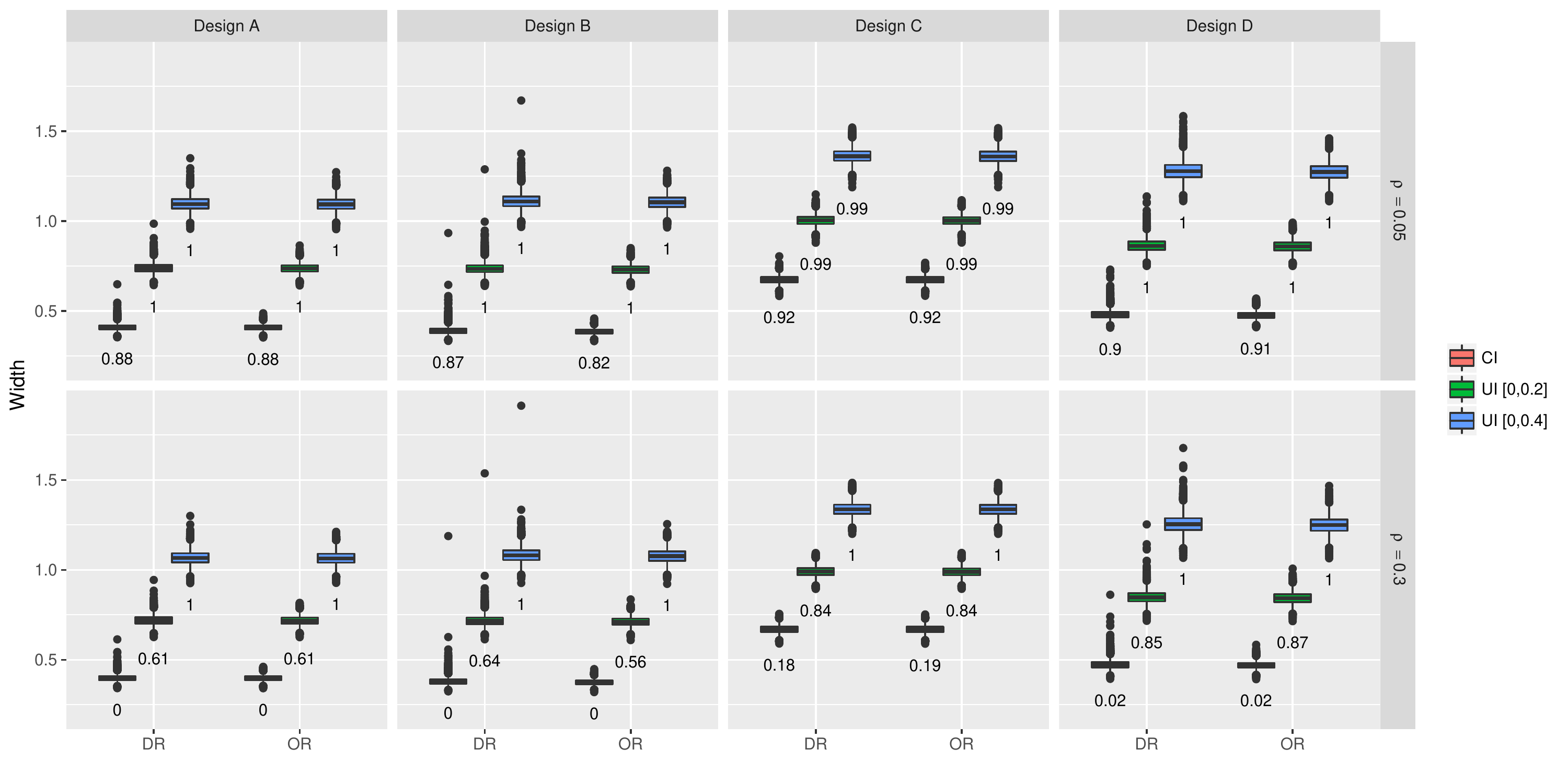}
\caption{Treatment assignment with low L1.}
\label{t.p1.500.fig}
\end{subfigure}
\begin{subfigure}[b]{\textwidth}
\includegraphics[width=\textwidth]{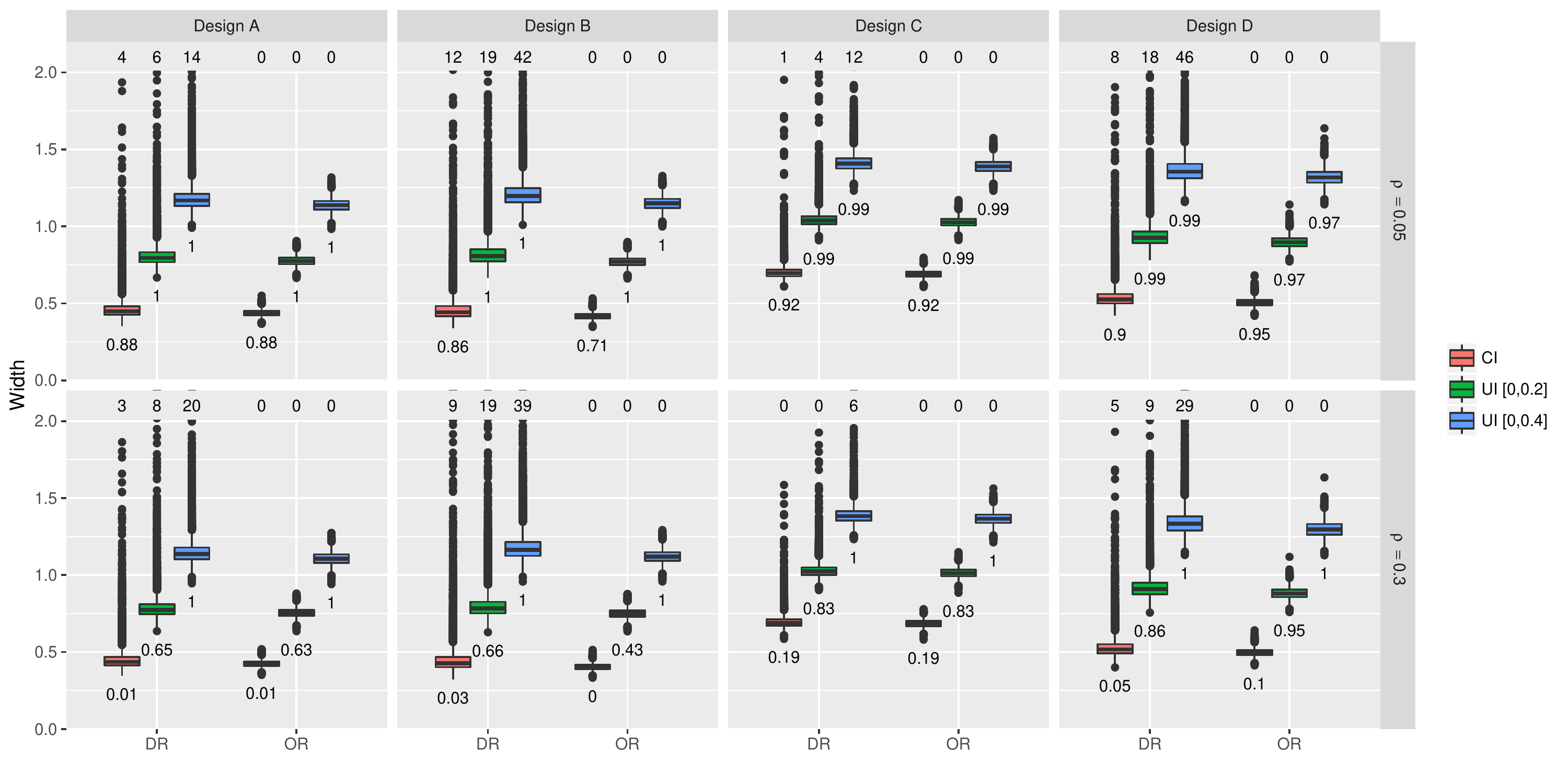}
\caption{Treatment assignment with high L1.}
\label{t.p2.500.fig}
\end{subfigure}
\caption{Boxplot of the width of two 95\% uncertainty intervals (assuming $\rho_j \in [0,0.2]$, green, and $\rho_j \in [0,0.4]$, blue, $j=0,1$) and the 95\% confidence interval, red, for the doubly robust (DR) and outcome regression (OR) estimator of $\tau$ under design A-D for $\rho_0=\rho_1=0.05$ and $0.3$, with sample size 500. The empirical coverage of each interval is written below each boxplot and the number of outliers that lie outside the window is written at the top of the window above each boxplot.}
\label{t.500.fig}
\end{figure}

\section{Effect of regular food intake on health}
SHARE is a longitudinal survey on health, socio-economic status and social networks of individuals aged 50 years or older from several European countries (\citealp{SHARE}). The sampling in SHARE is on a household level where all residents in the household (almost exclusively one individual or one man and one women) are interviewed. We focus this study on women in the 13 countries that participate in both wave 4 and 5 of SHARE, which were collected in 2011 (baseline) and 2013 (follow-up). The observed sample consists of 12 842 individuals. We are interested in investigating the causal effect of regular food intake on health. We define regular food intake as eating at least 3 full meals a day at baseline. A full meal is defined as eating more than 2 items or dishes when you sit down to eat. For example, eating potatoes, vegetables, and meat; or eating an egg, bread, and fruit are both considered full meals.  The health outcome used is change in maximum grip strength (in kg, maximum of 4 measures using a dynamometer) from 2011 to 2013. Grip strength is associated with both health-related quality of life, disability and mortality, see e.g. \cite{GS:Health} and \cite{GS:Death}. 

In order to estimate the causal effect of interest we control for covariates measured at baseline. These covariates include health, cognition, lifestyle, and socioeconomic variables as well as other background characteristics. The health variables include self reported health (excellent; very good; good; fair; or poor), number of problems with mobility (such as walking; lifting small objects; lifting heavy objects; etc., maximum 10), number of chronic diseases (such as diabetes; cancer; asthma; etc., maximum 15), depression (number of symptoms of depression, maximum 12, using the EURO-D scale), body mass index ($\mbox{kg/m}^2$) and limitations in daily life due to health problems (yes; no). We measure cognition with the number of animals the subject was able to state during 1 minute. The lifestyle variables consist of high alcohol use (drinking at least one glass of alcohol for women and two glasses for men at least 5 days a week), smoking (smoker; stopped smoking; non-smoker), physical inactivity (if respondents engaged in moderate to vigorous physical activity at most 1 to 3 times a month) and having a social network (have someone to discuss important things with, talk at least several times a week). The socioeconomic variables include education level (level 0-1; 2; 3; 4; or 5-6, using ISCED-97 scale) and whether or not the subject is living in an apartment or freestanding building. Finally, demographic characteristics consist of age, sex and country of residence (Austria; Germany; Sweden; Netherlands; Spain; Italy; France; Denmark; Switzerland; Belgium; Czech republic; Slovenia; or Estonia).

We estimated the causal effects by controlling for all main effects in the two potential outcome models. We used two different treatment assignment models, one with all main effects and one more flexible. The flexible treatment assignment model was fitted with a LASSO to select terms from all main effects together with interactions and quadratic terms. More specifically we use the R package {\tt glmnet} and  choose the largest value of the tuning parameter such that the mean cross validated error is within one standard error of the minimum, see \cite{glmnet} for details. The models including the selected terms are then refitted using maximum likelihood. The balance in the propensity scores is fairly similar between the main terms and LASSO based treatment assignment model, see Figure \ref{prop.sc.hist}.

\begin{figure}
\caption{Overlaid histograms showing the amount of imbalance in the propensity scores between treated and untreated for the two different sets of covariates, main effects (left) and LASSO (right). This bin size is default from the function {\tt hist} in {\tt R} statistical software and this bin size was also used to derive L1 (0.23 for main effects and 0.25 for LASSO).}
\begin{center}
\includegraphics[width=0.7\textwidth]{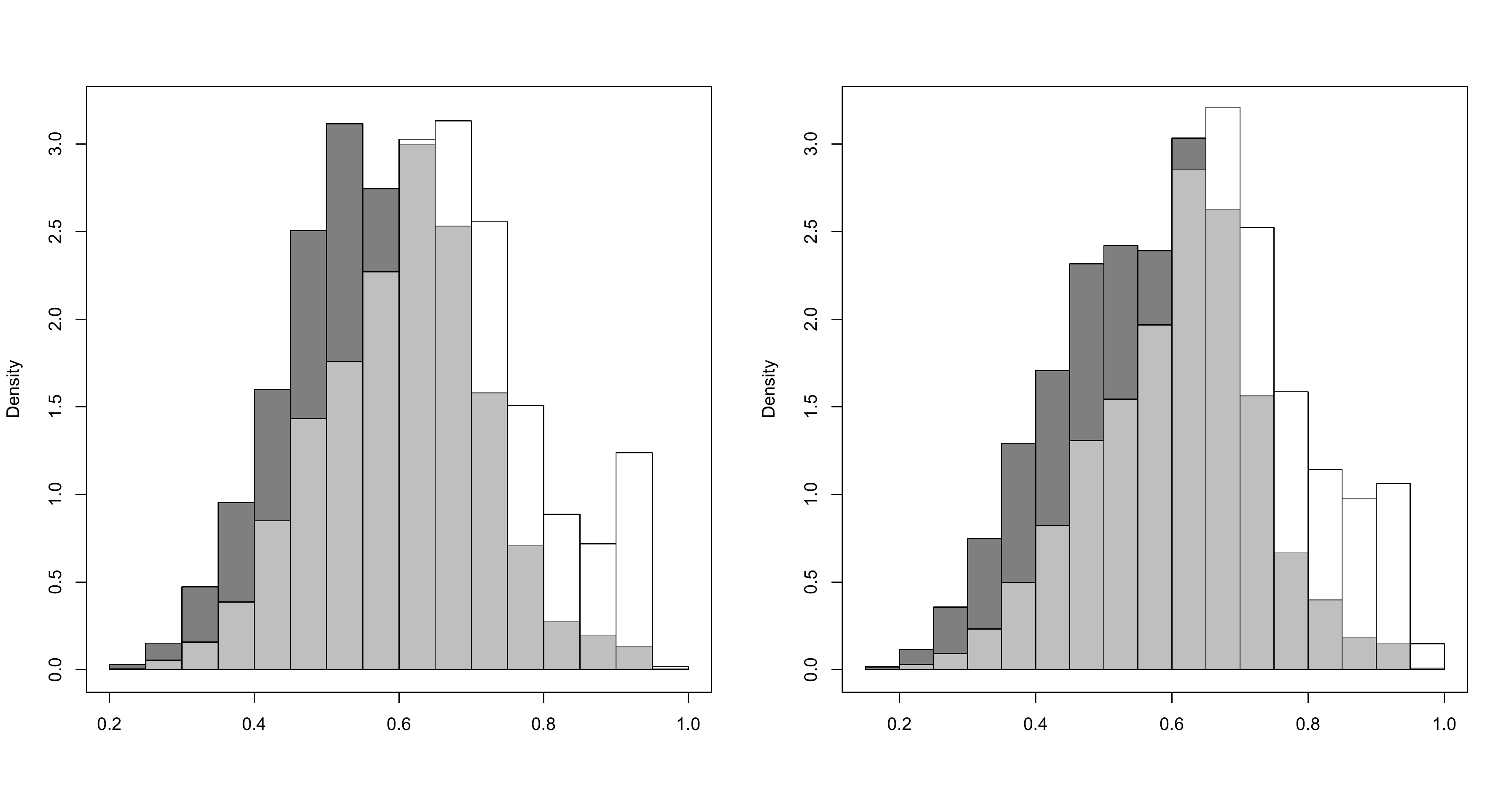}
\end{center}
\label{prop.sc.hist}
\end{figure}

\begin{table}[ht]
\centering
\caption{Estimates of the effect of regular food intake on change in grip strength, using doubly robust and outcome regression estimators. The confidence intervals are derived assuming ignorability of treatment assignment. The uncertainty intervals are derived assuming that $\rho \in [-0.02, 0.02]$. The estimates are derived with two different treatment assignment models, including main effects (top) and using LASSO as selection method for main effects, interactions and quadratic terms (bottom).}
\begin{tabular}{clllll}
 && coef & CI&UI, $|\rho|\leq0.02$ \\ 
  \cline{2-5}
\multirow{5}{*}{\rot{Main effects}}&$\tau_{OR}^1$ & 0.28 & (0.07, 0.48)&(-0.11, 0.66)\\ 
&$\tau_{DR}^1$  & 0.26 & (0.05, 0.46)&(-0.12, 0.64)\\ 
&$\tau_{OR}$ & 0.27 & (0.08, 0.46) &(-0.09, 0.63)\\ 
&$\tau_{DR}$ & 0.26 & (0.06, 0.45)&(-0.11, 0.62)\\ 
\multirow{4}{*}{\rot{LASSO}}&$\tau_{OR}^1$ & 0.28 & (0.07, 0.48)&(-0.10, 0.65)\\ 
&$\tau_{DR}^1$  & 0.24 & (0.03, 0.45)&(-0.15, 0.63)\\ 
&$\tau_{OR}$ & 0.27 & (0.08, 0.46) &(-0.09, 0.63) \\ 
&$\tau_{DR}$ & 0.25 & (0.05, 0.44)&(-0.12, 0.62)\\ 
  \cline{2-5}
\end{tabular}
\label{ex.tab}
\end{table}

In Table \ref{ex.tab} we can see that the estimates of $\tau$ and $\tau_1$ assuming ignorability of treatment assignment, are significant (95\% CI do not cover zero) for all estimators and estimated to between $0.24$ and $0.28$, which can be compared to $0.85$, the average decrease in maximum grip strength of the total study sample. All estimates obtained are fairly similar, in particular when compared to the extra variation introduced by the uncertainty in unobserved confounding (compare UIs with CIs). Indeed, the uncertainty intervals assuming $\rho \in [-0.02, 0.02]$ contain 0. The bounds $\max|\rho|=0.02$ corresponds to unobserved confounding explaining, e.g., 2 \% of the unexplained variation in the outcome models and the treatment assignment models (see interpration of $\rho$ given in Section \ref{Identif.sec} above). We have no reason to believe that such unobserved confounding is unreasonable. Thus, here, taking into account uncertainty in unobserved confounding yields inconclusive results, i.e. the data does not give us evidence for a positive effect in contrast with the naive conclusion that would typically be taken by only considering sampling uncertainty through classical confidence intervals. 

Note, finally, that this analysis has been performed assuming dropout at follow up to be ignorable. Non-ignorable dropout if suspected could be dealt with similarly by introducing a new bias parameter \citep{Genback:2015}, thereby further increasing even more the uncertainty around the estimates obtained.


\section{Discussion}
Causal inference from observational data is often based on the assumption of no unobserved confounding variables. This identifying assumption is typically not empirically testable without further assumptions and/or information such as, e.g., the known existence of instrumental variables \citep{LunaJ:2014}. This paper proposes an inferential approach for outcome regression and doubly robust estimators that takes into account uncertainty on the possible existence of unobserved confounding. The method proposed is computationally fast and easy to apply (the \texttt{R}-package \texttt{ui} is available at \texttt{http://stat4reg.se/software}). 

Outcome regression and double robust estimators make model assumptions which, if mistaken, also imply bias.  On the other hand, model misspecification can in principle be empirically investigated. In the simulated settings, even though the model misspecification was quite severe for the outcome regression estimator, bias due to unobserved confounding dominated model misspecification bias when $\rho \geq 0.1$ and even more so when propensity scores were not too close to zero or one. More generally, while model misspecification can under some assumptions be made arbitrarily small asymptotically (by increasing model complexity), bias/uncertainty due to unobserved confounding remains unchanged and therefore more relevant when increasing sample size.

We have focused on misspecification of the outcome models instead of the treatment assignment model. The latter is not only more challenging theoretically, but more importantly, one can argue that model building is less difficult for the treatment assignment model than for outcome models since for the former all data is available and no extrapolation is performed, while outcome models are fitted only on one sub-sample at a time (e.g., the controls) and are used to extrapolate on the other sub-sample (e.g., the treated). Extrapolations are thus done on part of the sample space which is sparsly populated, hence, where the model specification is difficult to check. Yet, it has been shown that, for double robust estimators, mild misspecification of both models (for treatment assignment and outcome) may lead to large bias in specific situations, in which case regression outcome estimation may be preferable \citep{kang:07}, or improved versions of the classic double robust estimator used here; see \cite{RotnitzkyStijn:15} for a review.   
  
The proposed uncertainty intervals can be used to perform a sensitivity analysis. For example, for all the estimators presented in Table \ref{ex.tab}, the UIs would approximately be bounded below by zero if constructed using $\rho\in [-0.01,0.01]$. Thus, the 5\% significance conclusion is here sensitive to unobserved confounding of magnitude $\max|\rho|\geq 0.01$. However, our experience is that sensitivity analyses are difficult to communicate to the layman for whom statistical hypothesis testing may already be a difficult concept. 
 We therefore advocate here the more intuitive interval estimation approach, i.e. providing an UI for the effect of interest given some a priori assumption on unobserved confounding and a desired coverage level.

\section*{Acknowledgements}

We are grateful to Elena Stanghellini, Arvid Sj\"olander, Anders Lundquist and Anita Lindmark for helpful comments. This work was supported by the Swedish Research Council for Health, Working Life and Welfare and the Marianne and Markus Wallenberg Foundation.\vspace*{-8pt}


  \bibliographystyle{biom} 
 \bibliography{referenser}


\section*{Appendix A}

\begin{Lem}
\label{Lem1}
Under Assumption \ref{Ass1b}, \ref{Ass2b}, \ref{Ass3} and \ref{Ass4} the bias of the ordinary least squares estimate of $\beta^j$, given $\rho_j$ is:  
\begin{equation*}\label{distortion.eq}
E(\hat\beta^j_{OLS})=\beta^j+ (-1)^{1+j}  \rho_j \sigma_{j} E((\X_j' \X_j)^{-1} \X_j' {\boldsymbol \lambda_j } ) ,
\end{equation*}
for $j=0,1$; where ${\boldsymbol \lambda}_j = [\lambda_j(g(x_1; \gamma)), \cdots , \lambda_j(g(x_{n_j}; \gamma))]'$, $\lambda_0(g(x_i; \gamma))= \frac{\phi(g(x_i; \gamma))}{1-\Phi(g(x_i; \gamma))}$,  $\lambda_1(g(x_i; \gamma))=\frac{\phi(g(x_i; \gamma))}{\Phi(g(x_i; \gamma))}$. The proof follow from $E(y^j|x,z=j)=E(y^j|x)+ (-1)^{1+j}\rho_j\sigma_j \lambda_j(g(x; \gamma))$ and can be found for $z=1$ in \cite{Genback:2015}, the proof when $z=0$ is similar.
\end{Lem}

\begin{subtheorem}{RegAss}
\label{RegAss1}
\begin{RegAss}
\label{RegAss1a}
There exist a constant $c$ such that $E\left(\frac{1}{( \frac{1}{n}\sum_i z_i)^2}\right)<c<\infty, \forall n$ and $E(z)>0.$
\end{RegAss}
\begin{RegAss}
\label{RegAss1b}
$\left|E(z_i y_i)\right|<\infty$ and $\left|E(\hat\beta^{0'}_{OLS} z_i x_i)\right|<\infty$.
\end{RegAss}
\end{subtheorem}

\begin{Lem}
\label{Lem2}
Under Regularity Assumption \ref{RegAss1a}:
$$\limes E\left(\frac{1}{n_1}\sum_{i=1}^{n}z_i f(x_i,y_i)\right) = \frac{1}{E(z_i)}E\left(z_i f(x_i,y_i)\right),$$ 
for any function $f$ of $x_i$ and $y_i$ such that $\left| E\left(z_i f(x_i,y_i)\right)\right|<\infty$.
\end{Lem}
	\begin{proof}
 $$\limes E\left(\frac{1}{n_1}\sum \limits_{i =1}^nz_i f(x_i,y_i)\right)=\limes  E\left(\frac{1}{\frac{1}{n}\sum_i z_i}\right) E\left(\frac{1}{n}\sum \limits_{i =1}^nz_i f(x_i,y_i)\right)=\frac{1}{E(z_i)}E\left(z_i f(x_i,y_i)\right).$$
The first equality follow from:
$$
\Cov\left(\frac{1}{\sum_j z_j}, \sum z_i f(x_i,y_i)\right)= \begin{cases}\frac{1}{n} E\left(z_i f(x_i,y_i)\right) E\left(\frac{-1}{\frac{1}{n^2}(\sum_i z_i))(\sum_i z_i+1)}\right) \mbox{if }z_i=0\\
0\mbox{ if } z_i=1.\end{cases}
$$

if there exist a constant $c$ such such that $E\left(\frac{1}{( \frac{1}{n}\sum_i z_i)^2}\right)<c<\infty$ then:
$$
\limes \left| \Cov\left(\frac{1}{\sum_j z_j}, \sum z_i f(x_i,y_i)\right) \right| \leq \limes\left| E\left(z_i f(x_i,y_i)\right)\right| E\left(\frac{1}{(\frac{1}{n}\sum_iz_i)^2} \right)\frac{1}{n}=0.
$$
The second equality follow from: $\frac{1}{n}\sum_i z_i\xrightarrow{p}E(z)$ (weak law of large numbers) and  $\frac{1}{\frac{1}{n}\sum_i z_i}\xrightarrow{p}\frac{1}{E(z)}$ if $E(z)>0$ (continuous mapping theorem). By dominated convergence theorem $\limes E\left(\frac{1}{\frac{1}{n}\sum_i z_i}\right)=\frac{1}{E(z)}$ (since $E\left(\frac{1}{| \frac{1}{n}\sum_i z_i|}\right)<\sqrt{c}<\infty$), the theorems used can be found for instance in \cite{Rosenthal}. 
	\end{proof}

To calculate the bias of the doubly robust estimator we need a regularity assumption to be able to exchange $\hat p(x)$ with $p(x)$ using the uniform integrability convergence theorem (\citealp{Rosenthal}).

\begin{RegAss}
\label{RegAss2} 
$\sum \limits_{i =1}^n z_i \frac{f(x_i,y_i)-\hat\beta^{1'}_{OLS}x_i}{\hat{p}(x_i)}$, and $\sum \limits_{i =1}^n (1-z_i)\frac{f(x_i,y_i)-\hat\beta^{0'}_{OLS}x_i}{1-\hat{p}(x_i)}$ are uniformly integrable $\forall x \in \mathcal{X}$, the support of $x.$
\end{RegAss}

\section*{Appendix B}
\subsection*{Proof of Proposition \ref{Prop.OR1}}

Under, Assumption \ref{Ass1b}, \ref{Ass3}, \ref{Ass4a} and Regularity Assumption \ref{RegAss1}: 
\begin{align*}
\limes \mbox{bias}_T(\hat \tau_{OR}^1)&=\limes E(\hat\tau^1_{OR})- \tau^1=\limes E\left(\frac{1}{n_1}\sum \limits_{i =1}^nz_i y_i-\hat\beta^{0'}_{OLS}\frac{1}{n_1}\sum_{i =1}^{n} z_i x_i \right) - \tau^1\\
&= E(y^1\mid z=1) - E\left(\hat\beta^{0'}_{OLS} z_i x_i \right) \frac{1}{E(z)} - \tau^1\\
&= E(y^0\mid z=1)-E(\hat\beta^{0'}_{OLS})E(x\mid z=1)  \\
&=E(y^0\mid z=1)-\beta^{0'} E(x\mid z=1) +  \rho_0 \sigma_0 E((\X_0'\X_0)^{-1}\X_0' \boldsymbol \lambda_0) E(x\mid z=1) \\
&=E(y^0\mid z=1)-(E(y^0\mid z=1)-\rho_0\sigma_0 E(\lambda_1(g(x; \gamma)) \mid z=1)) \\
&\qquad\qquad\qquad\qquad+  \rho_0 \sigma_0 E((\X_0'\X_0)^{-1}\X_0' \boldsymbol \lambda_0) E(x\mid z=1)  \\
&= \rho_0\sigma_0 \left( E(\lambda_1(g(x; \gamma))\mid z=1)+ E((\X_0'\X_0)^{-1}\X_0' \boldsymbol \lambda_0) E(x\mid z=1)\right). 
\end{align*}
The equality between line 1 and 2 follow from Lemma 2. The equality between line 3 and 4 follow from Lemma 1 and the equality between line 4 and 5 follow from the fact that $E(y^0\mid z=1)=E(\beta^{0'}x_i +\rho_0\sigma_0\lambda_1(g(x_i; \gamma))\mid z_i=1)=\beta^{0'}E(x_i\mid z_i=1) +\rho_0\sigma_0E(\lambda_1(g(x_i; \gamma))\mid z_i=1)$. 

Since Assumption \ref{Ass4} is fulfilled (the regression model is correctly specified) $\bias_M(\hat \tau_{OR}^1)=0$, and therefore $\bias_T(\hat \tau_{OR}^1)=\bias_C(\hat \tau_{OR}^1)$. \qed

\vspace{5mm}
\noindent Under Assumption \ref{Ass1b}, \ref{Ass2b}, \ref{Ass3} and \ref{Ass4}: 
\begin{align*}
\mbox{bias}_T(\hat \tau_{OR})&=E\left[\frac{1}{n}\sum \limits_{i =1}^n\left( \hat\beta^{1'}_{OLS}x_i-\hat\beta^{0'}_{OLS}x_i\right)\right] -\tau=E\left[\frac{1}{n}\left({\bf1}_n \X \hat\beta^{1'}_{OLS}- {\bf1}_n \X \hat\beta^{0'}_{OLS}\right)\right] - \tau\\
&=\frac{1}{n}E\left[ {\bf1}_n \X \left(  \beta^1 - \beta^0 \right) +  {\bf1}_n \X  \left(\rho_1\sigma_1(\X_1'\X_1)^{-1}\X_1' \boldsymbol\lambda_1 +\rho_0\sigma_0(\X_0'\X_0)^{-1}\X_0' \boldsymbol \lambda_0 \right)\right]-\tau\\
&= \frac{1}{n}E\left[ {\bf1}_n \X  \left(\rho_1\sigma_1(\X_1'\X_1)^{-1}\X_1' \boldsymbol\lambda_1 +\rho_0\sigma_0(\X_0'\X_0)^{-1}\X_0' \boldsymbol \lambda_0 \right)\right].
\end{align*}
Where the equality between line 1 and 2 follow from Lemma 1. Since Assumption \ref{Ass4} is fulfilled (the regression model is correctly specified) $\bias_M(\hat \tau_{OR})=0$, and therefore $\bias_T(\hat \tau_{OR})=\bias_C(\hat \tau_{OR})$. \qed

\subsection*{Proof Proposition \ref{Prop.OR2}}
Under, Assumption \ref{Ass1b}, \ref{Ass3} and Regularity Assumption \ref{RegAss1}: 
\begin{align*}
\limes \bias_T(\hat \tau^1_{OR}) &=E(y^0|z=1)  -E(\hat\beta^{0'}_{OLS}) E(x |z=1)\\
&=E[ E(y^0|x,z=1) |z=1] \\
&\qquad\qquad- E\left[(\X_0'\X_0)^{-1}\X_0' f^0(x) - \rho_0 \sigma_0(\X_0'\X_0)^{-1}\X_0'  {\boldsymbol \lambda_0 } \right] E(x|z=1) \\
&=E\left[ f^0(x) + \rho_0\sigma_0 \lambda_1(g(x; \gamma)) |z=1\right] \\
&\qquad\qquad- E\left[(\X_0'\X_0)^{-1}\X_0' f^0(x) - \rho_0 \sigma_0(\X_0'\X_0)^{-1}\X_0'  {\boldsymbol \lambda_0 } \right] E(x|z=1) \\
&=E\left[ f^0(x) |z=1\right] - E\left[(\X_0'\X_0)^{-1}\X_0' f^0(x) \right] E(x|z=1) +\mbox{bias}_C(\tau_{OR}^1)\\
&=\mbox{bias}_M(\tau_{OR}^1) +\mbox{bias}_C(\tau_{OR}^1),
\end{align*}
where the first equality can be seen in proof of Proposition \ref{Prop.OR1}.  \qed

\vspace{5mm}
\noindent Under, Assumption \ref{Ass1b}, \ref{Ass2b} and \ref{Ass3}: 
\begin{align*}
\bias_T(\hat \tau_{OR})&=E\left[\frac{1}{n}\sum \limits_{i =1}^n\left( \hat\beta^{1'}_{OLS}x_i-\hat\beta^{0'}_{OLS}x_i\right)\right] - \tau \\
&=E\left[\frac{1}{n}\left({\bf1}_n \X \hat\beta^{1'}_{OLS}- {\bf1}_n \X \hat\beta^{0'}_{OLS}\right)\right] - \tau\\
&= E\left[\frac{1}{n}E\left( \left.{\bf1}_n \X (\X_1'\X_1)^{-1}\X_1' \boldsymbol \y_1 - {\bf1}_n \X (\X_0'\X_0)^{-1}\X_0' \boldsymbol \y_0\right| \X\right)\right] -\tau\\
&=\frac{1}{n}E\left[ {\bf1}_n \X \left((\X_1'\X_1)^{-1}\X_1'E\left( \left.  \boldsymbol \y_1 \right| \X\right) -(\X_0'\X_0)^{-1}\X_0'E\left( \left.  \boldsymbol \y_0\right| \X\right) \right)\right]- \tau\\
&=\frac{1}{n}E\left( {\bf1}_n \X \left[(\X_1'\X_1)^{-1}\X_1' \left(  \boldsymbol f^1(x)_{|z=1} +\rho_1 \sigma_1 {\boldsymbol \lambda_1}\right)\right]\right)\\
&\qquad\qquad\qquad\qquad-\frac{1}{n}E\left( {\bf1}_n \X \left[(\X_0'\X_0)^{-1}\X_0'\left(  \boldsymbol f^0(x)_{|z=0} +\rho_0 \sigma_0 {\boldsymbol \lambda_0}\right)\right]\right)- \tau\\
&=\mbox{bias}_C(\tau_{OR})+\frac{1}{n}E\left[ {\bf1}_n \X \left((\X_1'\X_1)^{-1}\X_1'  \boldsymbol f^1(x)_{|z=1}  \right) \right]\\
&\qquad \qquad \qquad \qquad -\frac{1}{n}E\left[ {\bf1}_n \X \left((\X_0'\X_0)^{-1}\X_0'  \boldsymbol f^0(x)_{|z=0} \right) \right]- \tau\\
&=\mbox{bias}_C(\tau_{OR})+\frac{1}{n}E\left[ {\bf1}_{n_0} \X_0 (\X_1'\X_1)^{-1}\X_1'  \boldsymbol f^1(x)_{|z=1} \right]- E\left((1-z) f^1(x)\right)\\
&\qquad \qquad \qquad \qquad + E\left(z f^0(x) \right) - \frac{1}{n}E\left[ {\bf1}_{n_1} \X_1 (\X_0'\X_0)^{-1}\X_0'  \boldsymbol f^0(x)_{|z=0} \right],
\end{align*}
where the last equality follow from: $\frac{1}{n}E\left[ {\bf1}_{n_0} \X_0 (\X_0'\X_0)^{-1}\X_0'  \boldsymbol f^0(x)_{|z=0} \right]=E((1-z)f^0(x))$ and $\frac{1}{n}E\left[ {\bf1}_{n_1} \X_1 (\X_1'\X_1)^{-1}\X_1'  \boldsymbol f^1(x)_{|z=1} \right]=E(zf^1(x))$.\qed

\subsection*{Proof Proposition \ref{Prop.DR}}
Under, Assumption \ref{Ass1b}, \ref{Ass2b}, \ref{Ass3}, and Regularity Assumption \ref{RegAss1} and \ref{RegAss2}: 
\begin{align*}
\limes \mbox{bias}_C(\hat \tau^1_{DR})&=\limes E(\hat\tau^1_{DR}) - \tau^1\\
&=\limes \left[\frac{1}{n_1}\sum  \limits_{i =1}^n z_i\left(y_i-\hat\beta^{0'}_{OLS}x_i\right)-\frac{1}{n_1}\sum  \limits_{i =1}^n(1-z_i)\frac{y_i-\hat\beta^{0'}_{OLS}x_i}{1-\hat p(x_i)} \right]- \tau^1\\
&=\limes \left[\frac{1}{n_1}\sum  \limits_{i =1}^n z_i\left(y_i-\hat\beta^{0'}_{OLS}x_i\right)-\frac{1}{n_1}\sum  \limits_{i =1}^n(1-z_i)\frac{y_i-\hat\beta^{0'}_{OLS}x_i}{1-p(x_i)} \right]- \tau^1\\
&=E(y^0 -\hat\beta^{0'}_{OLS} x \mid z=1)\\
&\qquad \qquad \qquad \qquad  - \frac{1}{E(z)}E_x\left[\left.E\left((1-z)(y^0-\hat\beta^{0'}_{OLS} x) \right| x \right)\frac{1}{1-p(x)} \right]  \\
&=E(y^0 -\hat\beta^{0'}_{OLS} x \mid z=1) \\
&\qquad \qquad \qquad \qquad - \frac{1}{E(z)}E_x\left[\left.E\left((y^0-\hat\beta^{0'}_{OLS} x) \right| x, z=0 \right)\frac{p(z=0|x)}{1-p(x)} \right]  \\
&=E(y^0 -\hat\beta^{0'}_{OLS} x \mid z=1) - \frac{1}{E(z)}E_x\left[\left.E\left(y^0\right| x, z=0 \right) -\hat\beta^{0'}_{OLS} x \right]  \\
&=E(y^0 -\hat\beta^{0'}_{OLS} x \mid z=1) - \frac{1}{E(z)}E_x\left[E(y^0 | x) - \rho_0\sigma_0\lambda_0(g(x; \gamma))- \hat\beta^{0'}_{OLS}  x \right]  \\
&=E(y^0 -\hat\beta^{0'}_{OLS} x \mid z=1) - \frac{E(y^0 ) }{E(z)} +\frac{\rho_0\sigma_0\lambda_0(g(x; \gamma))}{E(z)}+\frac{E\left(\hat\beta^{0'}_{OLS} x\right) }{E(z)} \\
&=E(y^0 \mid z=1) -\frac{E(z\hat\beta^{0'}_{OLS} x)}{E(z)}- \frac{E(z y^0 ) +E((1-z) y^0 ) }{E(z)} \\
&\qquad \qquad \qquad \qquad +\frac{\rho_0\sigma_0\lambda_0(g(x; \gamma))}{E(z)}+\frac{E\left( \hat\beta^{0'}_{OLS}x \right) }{E(z)} \\
&=\frac{\rho_0\sigma_0E(\lambda_0(g(x; \gamma)))}{E(z)},
\end{align*}
using Lemma 2 and the fact that $\left.E\left(\hat\beta^{0'}_{OLS} \right| x, z=0 \right)= \hat\beta^{0'}_{OLS}= \left.E\left(\hat\beta^{0'}_{OLS} \right| x\right)
$, $E(z y^0)=E(y^0|z)E(z)$, and $E((1-z) y^0 ) =E((1-z)  \hat\beta^{0'}_{OLS}x)$.\qed

\vspace{5mm}
\noindent Under, Assumption \ref{Ass1b}, \ref{Ass2b}, \ref{Ass3} and Regularity Assumption \ref{RegAss2}: 
\begin{align*}
\limes \mbox{bias}_C(\hat \tau_{DR})&=\limes \left[E(\hat\tau_{DR}) \right]- \tau= \limes\left[ \mu_1-\mu_0 \right]- \tau\\ &=\rho_1\sigma_1\lambda_1(g(x; \gamma))+\rho_0\sigma_0\lambda_0(g(x; \gamma)),
\end{align*}
where
\begin{align*}
\limes \mu_1&=\limes \left[E(\hat\beta^{1'}_{OLS}x) +E\left(\frac{z(y^1-\hat\beta^{1'}_{OLS}x)}{\hat p(x)}\right)\right] \\
&= E(\hat\beta^{1'}_{OLS}x) +E\left[ E\left(\left.\frac{z(y^1-\hat\beta^{1'}_{OLS}x)}{p(x)}\right| x \right)\right] \\
&=E(\hat\beta^{1'}_{OLS}x) +E\left[\frac{1}{p(x)} E\left(\left.z(y^1-\hat\beta^{1'}_{OLS}x)\right| x \right)\right] \\
&=E(\hat\beta^{1'}_{OLS}x) +E\left[\frac{1}{p(x)} E\left(\left.y^1-\hat\beta^{1'}_{OLS}x\right| x, z=1 \right)\Pr(z=1|x)\right] \\
&=E(\hat\beta^{1'}_{OLS}x) +E\left[E\left(\left.y^1\right| x, z=1 \right) -\hat\beta^{1'}_{OLS}x\right] \\
&=E(\hat\beta^{1'}_{OLS}x) +E\left[E\left(\left.y^1\right| x \right)+\rho_1 \sigma_1 \lambda_1(g(x; \gamma)) \right] -E(\hat\beta^{1'}_{OLS}x)\\
&=E\left(f^1(x) \right)+\rho_1 \sigma_1 E\left(\lambda_1(g(x; \gamma)) \right),
\end{align*}
and
\begin{align*}
\limes \mu_0&=\limes\left[E(\hat\beta^{0'}_{OLS}x) +E\left(\frac{(1-z)(y^0-\hat\beta^{0'}_{OLS}x)}{1-\hat p(x)}\right)\right] \\
&= E(\hat\beta^{0'}_{OLS}x) +E\left[ E\left(\left.\frac{(1-z)(y^0-\hat\beta^{0'}_{OLS}x)}{1-p(x)}\right| x \right)\right]\\
&=E(\hat\beta^{0'}_{OLS}x) +E\left[\frac{1}{1-p(x)} E\left(\left.(1-z)(y^0-\hat\beta^{0'}_{OLS}x)\right| x \right)\right] \\
&=E(\hat\beta^{0'}_{OLS}x) +E\left[\frac{1}{1-p(x)} E\left(\left.y^0-\hat\beta^{0'}_{OLS}x\right| x, z=0 \right)\Pr(z=0|x)\right] \\
&=E(\hat\beta^{0'}_{OLS}x) +E\left[E\left(\left.y^0\right| x, z=0 \right) -\hat\beta^{0'}_{OLS}x\right] \\
&=E(\hat\beta^{0'}_{OLS}x) +E\left[E\left(\left.y^0\right| x \right)-\rho_0 \sigma_0 \lambda_0(g(x; \gamma)) \right] -E(\hat\beta^{0'}_{OLS}x)\\
&=E\left(f^0(x) \right)-\rho_0 \sigma_0 E\left(\lambda_0(g(x; \gamma)) \right). 
\end{align*} \qed
\section*{Appendix C}
\subsection*{Variance of the outcome regression estimator}
The variance of the outcome regression estimator can be estimated by either the large sample variance or by the sandwich estimator since it is an m-estimator, see for instance \cite{Stefanski:2002} for details. In the simulations performed in this paper under Assumption \ref{Ass4}, both estimators have worked well. However, the sandwich estimator also performed well when Assumption \ref{Ass4} was not fulfilled (as expected), and hence it is the one we recommend and use for the simulations and the real data example. Also, the large sample variance for $\hat\tau_{OR}^1$ given below require the estimation of $\bbeta^1$ and $E\left(\mbox{Var}(y_i|x_i, z=1)\right)$, i.e. an additional model only used in the variance estimation, which we do not need for the sandwich estimator. 
\subsubsection*{Sandwich estimator}

\noindent Under Assumption \ref{Ass1} the variance of $\hat \tau_{OR}^1$ can be estimated by:
$$\widehat{\Var}(\hat{\tau}^1_{OR}) = (A_n^{-1} B_n (A_n^{-1})' )_{(1,1)}/ n$$
where (1,1) stand for the element on row 1 and column 1 and:

\begin{equation*}
A_n=\frac{1}{n}\sum \limits_{i =1}^n\left[
\begin{matrix}
z_i& z_i x_i  \\
\boldsymbol 0_{(p+1)\times1}&(1-z)x' x \end{matrix} \right], \quad B_n=\frac{1}{n}\sum \limits_{i =1}^n\Psi (y_i,  \boldsymbol{ \hat{\theta}})\Psi (y_i,  \boldsymbol{ \hat{\theta}})',
\end{equation*}
and
\begin{equation*}
\Psi (y_i,  \boldsymbol{ \hat{\theta}}) =\left[
\begin{matrix}
z_i(y_i- \hat \beta^0_{OLS}x_i -\tau^1) \\
(1-z_i)(y_i- \hat \beta^0_{OLS}x_i) x_i' 
\end{matrix} \right].
\end{equation*}
\vspace{2mm}

\noindent Under Assumption \ref{Ass1} and \ref{Ass2} the variance of $\hat \tau_{OR}$ can be estimated by:
$$\widehat{\Var}(\hat{\tau}_{OR}) = (A_n^{-1} B_n (A_n^{-1})' )_{(1,1)}/ n$$
where (1,1) stand for the element on row 1 and column 1 and:

\begin{equation*}
A_n=\frac{1}{n}\sum \limits_{i =1}^n\left[
\begin{matrix}
1& -x_i &x_i \\
\boldsymbol 0_{(p+1)\times1}&zx' x& \boldsymbol 0_{(p+1)\times(p+1)}\\
\boldsymbol 0_{(p+1)\times1}&\boldsymbol 0_{(p+1)\times(p+1)}& (1-z)x' x
 \end{matrix} \right],
\end{equation*}
\begin{equation*}
B_n=\frac{1}{n}\sum \limits_{i =1}^n\Psi (y_i,  \boldsymbol{ \hat{\theta}})\Psi (y_i,  \boldsymbol{ \hat{\theta}})',
\end{equation*}
and
\begin{equation*}
\Psi (y_i,  \boldsymbol{ \hat{\theta}}) =\left[
\begin{matrix}
(\hat \beta^1_{OLS}x_i- \hat \beta^0_{OLS}x_i) - \tau \\
z_i(y_i- \hat \beta^1_{OLS}x_i) x_i' \\
(1-z_i)(y_i- \hat \beta^0_{OLS}x_i) x_i' 
\end{matrix} \right].
\end{equation*}
\vspace{2mm}

\subsubsection*{Large sample variance}
The large sample variance of $\hat \tau_{OR}^1$ is given by:
$$\mbox{Var}(\hat \tau_{OR}^1) \simeq n_1^{-1} \left[2 \cdot E\left(\mbox{Var}(y_i|x_i, z=1)\right) +(\beta^1-\beta^0)\mbox{Cov}\left(x_i|z_i=1\right)(\beta^1-\beta^0)' \right],$$
which under Assumption \ref{Ass1}, \ref{Ass2} and \ref{Ass4} can be estimated by
$$\widehat\Var(\hat \tau_{OR}^1) \simeq n_1^{-1} \left[2 \left(\frac{\sum \limits_{i \in \mathcal{I}_1}e_i^2}{n_1-1} \right) +(\hat \beta^1-\hat \beta^0)\widehat{\Cov}\left(x_i|z_i=1\right)(\hat \beta^1-\hat \beta^0)' \right],$$
where $e_i$ and $\hat \beta^1$ are the residuals and coefficients estimates from the OLS regression of $\X_1$ on $\y_1$, and similarly $\hat \beta^0$ are the coefficients estimates from the OLS regression of $\X_0$ on $\y_0$.
 \begin{proof}
 Under Assumption \ref{Ass1} and \ref{Ass4a}:
\begin{align*}
\mbox{Var}(\hat \tau_{OR}^1) &\simeq \mbox{Var}(\hat \tau_{OR}^1|\z)\\
&=E\left.\left[\mbox{Var}\left(\frac{1}{n_1}\sum \limits_{i =1}^n\left(z_i( y_i -\hat\beta^{0'}_{OLS}x_i) \right)\right| x,  \z \right) \right] + \mbox{Var}\left[E\left(\hat \tau_{OR}^1 | x, \z \right) \right]\\
&=n_1^{-2} \left[ \sum \limits_{i =1}^n z_i \left( E\left(\mbox{Var}(y_i|x_i, z=1)\right) \right) +\mbox{Var}\left((\beta^1-\beta^0)\sum \limits_{i =1}^nz_ix_i\right)\right] \\ 
&+n_1^{-2} \left[ E\left( \sum \limits_{i =1}^n\sum \limits_{j =1}^n z_i z_j\mbox{Cov}(\hat\beta^{0'}_{OLS}x_i,\hat\beta^{0'}_{OLS}x_j | x, \z)\right)  \right] \\
&=n_1^{-1} \left[ E\left(\mbox{Var}(y_i|x_i, z=1)\right) +(\beta^1-\beta^0)\mbox{Cov}\left(x_i|z_i=1\right)(\beta^1-\beta^0)' \right]\\
&\qquad \qquad\qquad \qquad \qquad + n_1^{-2} E \left( \boldsymbol{1}_{1\times n_1}\X_1\mbox{Cov}(\hat\beta^{0}_{OLS}| x, \z) \X_1'\boldsymbol{1}_{n_1\times1}\right)\\
&=n_1^{-1} \left[2 \cdot E\left(\mbox{Var}(y_i|x_i, z=1)\right) +(\beta^1-\beta^0)\mbox{Cov}\left(x_i|z_i=1\right)(\beta^1-\beta^0)' \right],
\end{align*}
since
\begin{align*}
\sum \limits_{i =1}^n\sum \limits_{j =1}^n z_i z_j\mbox{Cov}(\hat\beta^{0'}_{OLS}x_i,\hat\beta^{0'}_{OLS}x_j | x, \z)  &= 
\sum \limits_{i =1}^n\sum \limits_{j =1}^n z_i z_j x_i'\mbox{Cov}(\hat\beta^{0}_{OLS}| x, \z) x_j\\
&=\boldsymbol{1}_{1\times n_1}\X_1\mbox{Cov}(\hat\beta^{0}_{OLS}| x, \z) \X_1'\boldsymbol{1}_{n_1\times1}\\
&=\boldsymbol{1}_{1\times n_1}\X_1\left[\mbox{Var}(y_i|x_i, z=1) \left(\X_1' \X_1\right)^{-1} \right] \X_1'\boldsymbol{1}_{n_1\times1}\\
&=n_1 \mbox{Var}(y_i|x_i, z=1).
\end{align*}
\end{proof}

\noindent The large sample variance of $\hat \tau_{OR}$ is given by:
\begin{align*}
\mbox{Var}(\hat \tau_{OR})&\simeq  n^{-2} \left(   E\left[ \boldsymbol{1}_{1\times n}\X \, \left[ \mbox{Cov}\left( \left. \hat\beta^{1}_{OLS} \right|  \X,\z \right) + \mbox{Cov}\left( \left. \hat\beta^{0}_{OLS} \right|  \X,\z \right)\right] \X' \, \boldsymbol{1}_{n\times 1} \right] \right)\\
&\qquad \qquad \qquad \qquad + n^{-1}(\beta^1-\beta^0)\mbox{Cov}(x_i)(\beta^1-\beta^0)',
\end{align*}
which under Assumption \ref{Ass1}, \ref{Ass2} and \ref{Ass4} can be estimated by:
\begin{align*}
\widehat{\Var}(\hat \tau_{OR})& =  n^{-2} \left(  \boldsymbol{1}_{1\times n}\X \, \left[\left(\frac{\sum \limits_{i \in \mathcal{I}_1}e_i^2}{n_1-1} \right)\X_1' \X_1 + \left(\frac{\sum \limits_{i \in \mathcal{I}_0}e_i^2}{n_0-1} \right)\X_0' \X_0\right] \X' \, \boldsymbol{1}_{n\times 1} \right)\\
&\qquad \qquad \qquad  + n^{-1}(\hat\beta^1-\hat\beta^0)\widehat\Cov(x_i)(\hat\beta^1-\hat\beta^0)'.
\end{align*}
\begin{proof}
\noindent Under Assumption \ref{Ass1} , \ref{Ass2}  and \ref{Ass4}:
\begin{align*}
\mbox{Var}(\hat \tau_{OR})&\simeq \mbox{Var}(\hat \tau_{OR}|\z)\\
&=E\left[ \mbox{Var}\left( \left. \frac{1}{n}\sum \limits_{i =1}^n\left( \hat\beta^{1'}_{OLS}x_i-\hat\beta^{0'}_{OLS}x_i\right) \right|  \X, \z \right) \right]\\
&\qquad \qquad \qquad \qquad + \mbox{Var}\left[ E\left( \left. \frac{1}{n}\sum \limits_{i =1}^n\left( \hat\beta^{1'}_{OLS}x_i-\hat\beta^{0'}_{OLS}x_i\right) \right|  \X, \z \right) \right] \\
&= n^{-2} \left(   E\left[ \mbox{Var}\left( \left.\sum \limits_{i =1}^n (\hat\beta^{1'}_{OLS}-\hat\beta^{0'}_{OLS})x_i \right|  \X, \z\right)   \right] \right) \\
&\qquad \qquad \qquad \qquad +n^{-1}(\beta^1-\beta^0)\mbox{Var}(x_i) (\beta^1-\beta^0)'  \\
&= n^{-2} \left(   E\left[ \boldsymbol{1}_{1\times n}\X \, \left[ \mbox{Cov}\left( \left. \hat\beta^{1}_{OLS} \right|  \X,\z \right) + \mbox{Cov}\left( \left. \hat\beta^{0}_{OLS} \right|  \X,\z \right)\right] \X' \, \boldsymbol{1}_{n\times 1} \right] \right)\\
&\qquad \qquad \qquad \qquad + n^{-1}(\beta^1-\beta^0)\mbox{Cov}(x_i)(\beta^1-\beta^0)'.
\end{align*}
\end{proof}

\subsection*{Variance of the double robust estimator}
To caluculate the variance of the doubly robust estimators we use the sandwich estimator since they both are m-estimators, see \cite{Stefanski:2002} for more details.

Under Assumption \ref{Ass1} the variance of $\hat \tau_{DR}^1$ can be estimated by:
$$\widehat{\Var}(\hat{\tau}^1_{DR}) = (A_n^{-1} B_n (A_n^{-1})' )_{(1,1)}/ n$$
where (1,1) stand for the element on row 1 and column 1 and:

\begin{equation*}
A_n=\frac{1}{n}\sum \limits_{i =1}^n\left[
\begin{matrix}
z_i& \left(z_i - \frac{(1-z_i)}{\Phi(\hat \bgamma x_i)}\right) x_i & (1-z_i)(y_i- \hat \beta^0_{OLS}x_i) \frac{\phi(\hat \bgamma x_i)}{\Phi(\hat \bgamma x_i)^2}x \\
\boldsymbol 0_{(p+1)\times1}&(1-z)x' x &\boldsymbol 0_{(p+1)\times(p+1)}\\
\boldsymbol 0_{(p+1)\times1}&\boldsymbol 0_{(p+1)\times(p+1)}& - \frac{d \Psi_3}{d \bgamma}
\end{matrix} \right],
\end{equation*}
\begin{equation*}
B_n=\frac{1}{n}\sum \limits_{i =1}^n\Psi (y_i,  \boldsymbol{ \hat{\theta}})\Psi (y_i,  \boldsymbol{ \hat{\theta}})',
\end{equation*}
and
\begin{equation*}
\Psi (y_i,  \boldsymbol{ \hat{\theta}}) =\left[
\begin{matrix}
\Psi_1  \\
\Psi_2  \\
\Psi_3 
\end{matrix} \right]=\left[
\begin{matrix}
z_i(y_i- \hat \beta^0_{OLS}x_i) - (1-z_i)\frac{y_i- \hat \beta^0_{OLS}x_i}{1-\Phi(\hat \bgamma x_i)} - z_i \tau^1 \\
(1-z_i)(y_i- \hat \beta^0_{OLS}x_i) x_i' \\
(z_i\frac{\phi(\hat \bgamma x_i)}{\Phi(\hat \bgamma x_i)}-(1-z)\frac{\phi(\hat \bgamma x_i)}{1-\Phi(\hat \bgamma x_i)})x_i'
\end{matrix} \right].
\end{equation*}
\vspace{2mm}

\noindent If Assumption \ref{Ass4} also holds some of the terms will converge to zero, i.e. the variance expression can be simplified for large samples. However, in the simulations performed for this paper, the simplified estimator have poor precision compared with the estimator given above. For the estimation of the standard error of $\hat \tau_{DR}$, however, a simplified version of the sandwich estimator performs rather good, hence this is what is proposed below.

 Under Assumption \ref{Ass1}, \ref{Ass2}, \ref{Ass4} and a correctly specified propensity score, the standard error of $\hat \tau_{DR}$ can be estimated by: (\citealp{Lunceford:2004})
$$\widehat{\Var}(\hat{\tau}_{DR}) = n^{-2} \sum \limits_{i =1}^n\hat{I}_i^2$$
\begin{displaymath}
\hat{I}_i= \hat\beta^{1'}_{OLS}x_i-\hat\beta^{0'}_{OLS}x_i
+z_i\, \frac{y_i-\hat\beta^{1'}_{OLS}x_i}{\hat{p}(x_i)}
-(1-z_i)\frac{y_i-\hat\beta^{0'}_{OLS}x_i}{1-\hat{p}(x_i)} -\hat{\tau}_{DR}.
\end{displaymath}

\section*{Appendix D}
\begin{figure} [htbp]
\centering
\begin{subfigure}[b]{\textwidth}
\includegraphics[width=\textwidth]{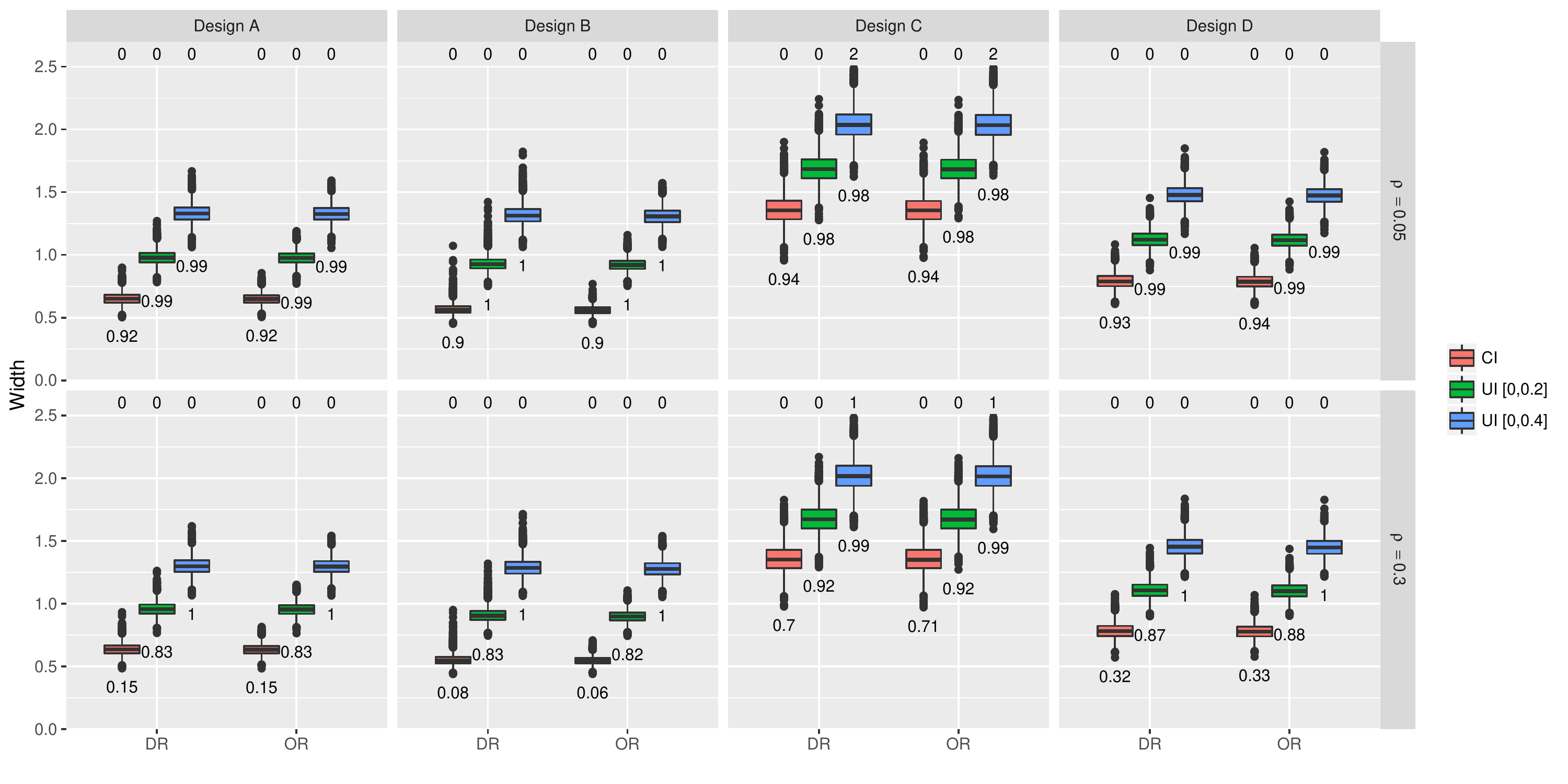}
\caption{Treatment assignment with low L1.}
\label{t1.p1.250.fig}
\end{subfigure}
\begin{subfigure}[b]{\textwidth}
\includegraphics[width=\textwidth]{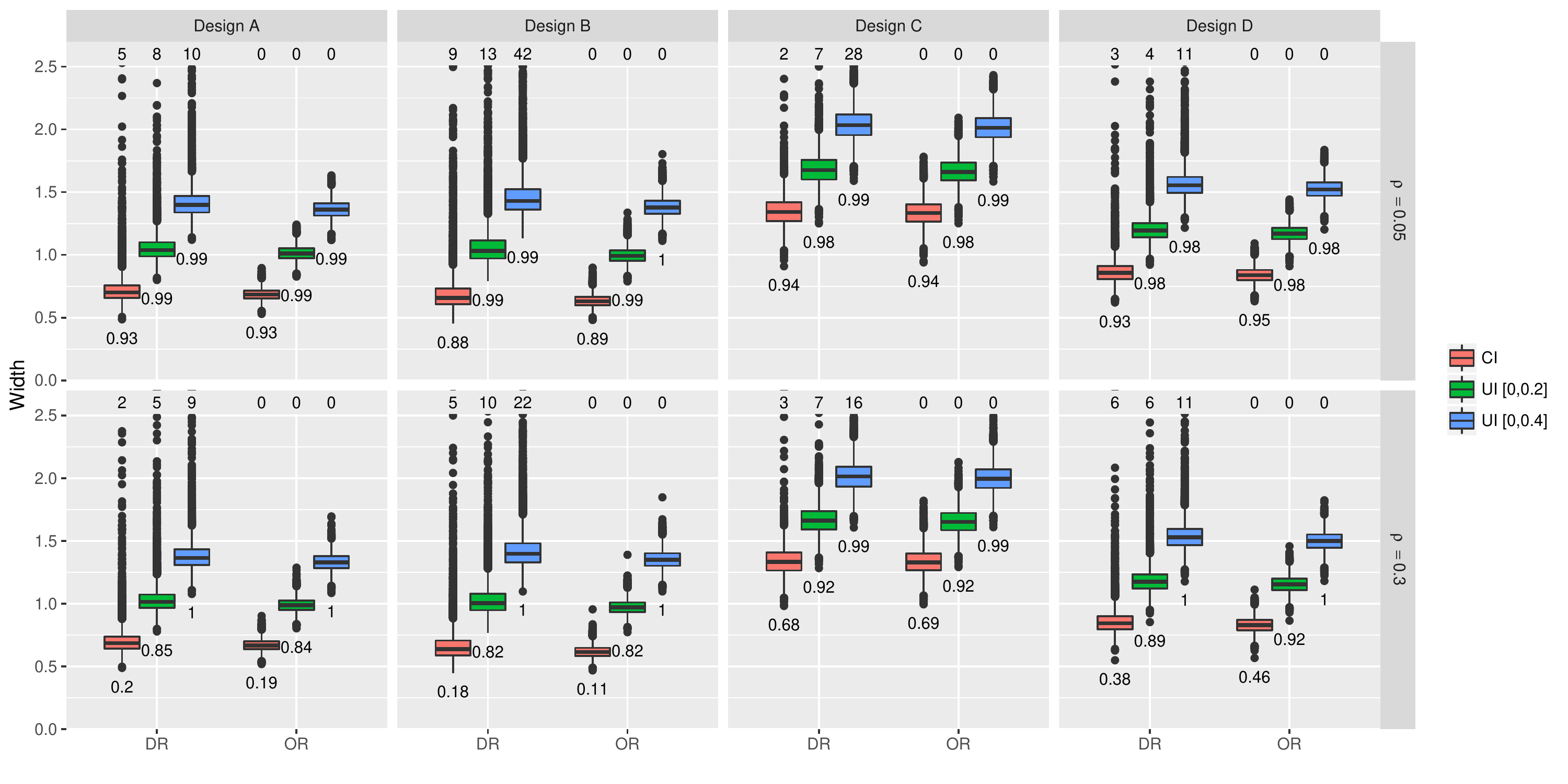}
\caption{Treatment assignment with high L1.}
\label{t1.p2.250.fig}
\end{subfigure}
\caption{Boxplot of the width of two 95\% uncertainty intervals (assuming $\rho_j \in [0,0.2]$, green, and $\rho_j \in [0,0.4]$, blue, $j=0,1$) and the 95\% confidence interval, red, for the doubly robust (DR) and outcome regression (OR) estimator of $\tau^1$ under design A-D for $\rho_0=\rho_1=0.05$ and $0.3$, with sample size 250. The empirical coverage of each interval is written below each boxplot and the number of outliers that lie outside the window is written at the top of the window above each boxplot.}
\label{t1.250.fig}
\end{figure}

\begin{figure} [htbp]
\centering
\begin{subfigure}[b]{\textwidth}
\includegraphics[width=\textwidth]{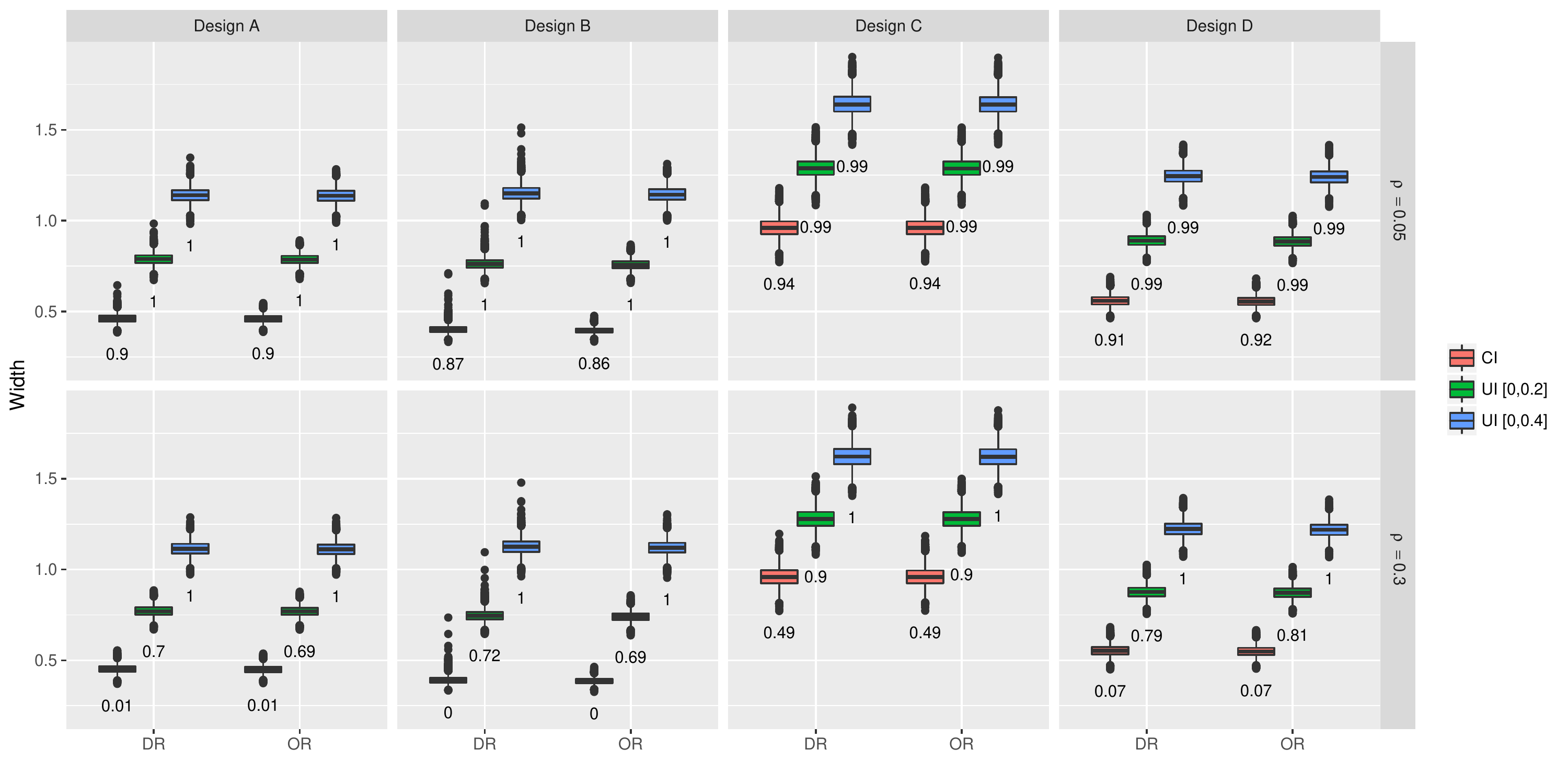}
\caption{Treatment assignment with low L1.}
\label{t1.p1.500.fig}
\end{subfigure}
\begin{subfigure}[b]{\textwidth}
\includegraphics[width=\textwidth]{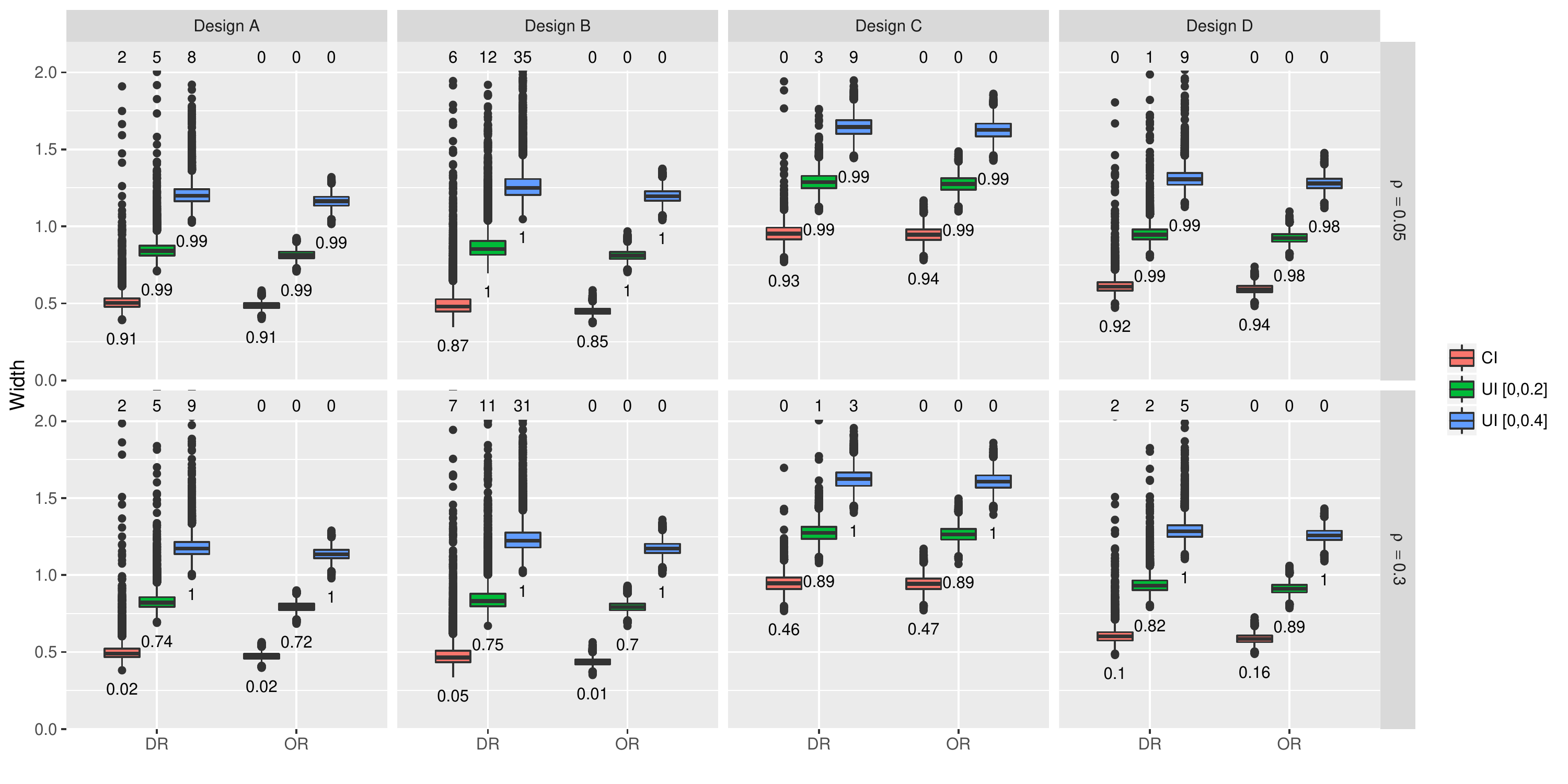}
\caption{Treatment assignment with high L1.}
\label{t1.p2.500.fig}
\end{subfigure}
\caption{Boxplot of the width of two 95\% uncertainty intervals (assuming $\rho_j \in [0,0.2]$, green, and $\rho_j \in [0,0.4]$, blue, $j=0,1$) and the 95\% confidence interval, red, for the doubly robust (DR) and outcome regression (OR) estimator of $\tau^1$ under design A-D for $\rho_0=\rho_1=0.05$ and $0.3$, with sample size 500. The empirical coverage of each interval is written below each boxplot and the number of outliers that lie outside the window is written at the top of the window above each boxplot.}
\label{t1.500.fig}
\end{figure}

\end{document}